\documentclass[pra,aps,twocolumn,longbibliography,nofootinbib,showpacs]{revtex4-1}

\usepackage[utf8]{inputenc}      
\usepackage[T1]{fontenc}         
\usepackage[english]{babel}      

\usepackage[colorlinks=true, citecolor=blue, linkcolor=orange, urlcolor=blue]{hyperref}

\usepackage{amsmath, amssymb, amsfonts, mathtools, bm, bbm}
\usepackage{amsthm}
\usepackage{physics}
\usepackage{braket}
\usepackage{esvect}
\usepackage{latexsym}
\usepackage{mathrsfs}
\usepackage[mathscr]{euscript}

\usepackage{graphicx}
\usepackage{wrapfig}
\usepackage{float}
\usepackage{adjustbox}
\usepackage{graphics, epstopdf}
\usepackage{epsfig}
\usepackage{tabularx}
\usepackage{booktabs}
\usepackage{dcolumn}

\usepackage{xcolor}
\usepackage{soul}
\usepackage[normalem]{ulem}
\usepackage{enumitem}
\usepackage{csquotes}
\usepackage{fmtcount}
\usepackage{verbatim}
\usepackage{appendix}
\usepackage{amstext} 
\let\oldEuScript\EuScript
\renewcommand{\EuScript}[1]{\text{\ensuremath{\oldEuScript{#1}}}}

\theoremstyle{plain}
\newtheorem{theorem}{Theorem}
\newtheorem{lemma}{Lemma}
\newtheorem{corollary}{Corollary}
\newtheorem{proposition}{Proposition}

\newtheorem{remark}{Remark}
\newtheorem{example}{Example}

\def\Tr{\operatorname{Tr}}

\def\bea{\begin{eqnarray}}
\def\eea{\end{eqnarray}}
\def\ba{\begin{array}}
\def\ea{\end{array}}

\def\beq{\begin{equation}}
\def\eeq{\end{equation}}

\renewcommand{\[}{\left[}

\usepackage{orcidlink}

\begin{document}



\title{Towards minimal conditions for ergotropy injection in  open quantum systems}

\author{Deependra Singh\,\orcidlink{0009-0009-0552-2506}} 
\email{deependrasingh@hri.res.in}

\author{Debarupa Saha\,\orcidlink{0009-0006-3226-7799}} 
\email{debarupa73@gmail.com}

\author{Ujjwal Sen\,\orcidlink{0000-0002-0091-5847}} 
\email{ujjwal@hri.res.in}

\affiliation{Harish-Chandra Research Institute, Chhatnag Road, Jhunsi,
Prayagraj 211 019, India\\
Homi Bhabha National Institute, Training School Complex,
Anushakti Nagar, Mumbai 400 094, India}

\begin{abstract}
Interactions between a quantum system and its environment can inject ergotropy into the system, raising the question of the minimal dimensionality and physical resources required for such injection. 
We first show that ergotropy injection is impossible under thermal operations when both the system and environment are qubits, whereas it becomes possible when the environment is enlarged to a qutrit. We further show that already in the qubit–qubit setting, relaxing environmental thermality and allowing interactions between system and environment allows ergotropy increment under energy-conserving unitaries. 
To elucidate the role of interactions, we consider a two-qubit isotropic XY interaction Hamiltonian and identify its distinct degeneracy regimes. We show that, in the central-block regime, when both the initial system and environmental states are incoherent, no ergotropic gain is possible when the environment is initially thermal, irrespective of the interaction strength. In contrast, environmental athermality in the form of population inversion, while retaining incoherence, enables ergotropic injection.
 We derive the optimal ergotropic gain and show that environmental coherence can enhance it, while system coherence alone need not be beneficial and can even reduce the gain. We further consider the double-degenerate regime, characterized by a finite interaction strength, and demonstrate positive ergotropic gain even for a thermal environment.

\end{abstract}
\maketitle
\section{Introduction}
Energy extraction from quantum systems is a central theme in quantum thermodynamics~\cite{Kosloff2013,Goold2016,Vinjanampathy2016,BinderBook2018,DeffnerCampbell2019}, with fundamental connections to both the operation of quantum batteries~\cite{AlickiFannes2013, Campaioli2024, Binder2015, Campaioli2017, Ferraro2018, Barra2022, GhoshMal2021, GhoshSen2022, Konar2022, Bhattacharyya2024prl,  Shukla2025, Chaki2025pra, Sarkar2025, Chaki2026jpa}
and the performance of microscopic heat engines~\cite{Scovil1959, Alicki1979, Kieu2004, AllahverdyanJohal2008, bhattacharjee2021, Biswas2022, Saha2026}.
For a quantum system in a given state, the maximum amount of energy that can be extracted by cyclic unitary driving is quantified by the ergotropy~\cite{Allahverdyan2004}, {which measures the departure of a state from passivity~\cite{Pusz1978, Lenard1978, Skrzypczyk2015,PerarnauLlobet2015passive, Brown2016, KSen2021}}. Quantum states with nonzero ergotropy, referred to as active states, can therefore serve as quantum batteries, providing a temporary storage of extractable energy that can subsequently be released to perform useful thermodynamic processing tasks. 

Beyond quantifying the amount of extractable energy, a central question has been to understand whether and how genuinely quantum features can enhance work extraction.  In particular, the roles of coherence, entanglement, and quantum correlations in determining the extractable work have been extensively investigated~\cite{Oppenheim2002, Hovhannisyan2013, Funo2013, PerarnauLlobet2015, Mukherjee2016, Korzekwa2016,
Francica2017, Alimuddin2019, Andolina2019, Francica2020, Touil2022,
SalviaGiovannetti2022, Ahuja2025, Mondal2025}.

Besides theoretical advances, recent experiments have begun to probe work extraction and energy storage in quantum systems across diverse platforms. Examples include superabsorption in an organic microcavity~\cite{Quach2022}, optimal charging of a superconducting qutrit quantum battery~\cite{Hu2022}, and the realization of star-topology quantum batteries using nuclear magnetic resonance techniques~\cite{Joshi2022},{photonic demonstrations of the charging process~\cite{Huang2023}, and the direct measurement of coherent
ergotropy in a single spin system~\cite{Niu2024}}. These experiments demonstrate the growing feasibility of implementing quantum energy-storage and extraction protocols in controllable physical platforms.

Quantum systems are generally not isolated and inevitably interact with their surrounding environments. Such interactions are commonly associated with decoherence and dissipation, but they can also increase the energetic usefulness of a quantum system by injecting ergotropy into it. This possibility has motivated a variety of mechanisms for charging quantum systems and manipulating stored ergotropy, including charger-mediated and dissipative protocols~\cite{Farina2019, Barra2019}, interaction-assisted local work extraction~\cite{Salvia2025}, collisional models with~\cite{Seah2021, Ciccarello2022}, environment-coherence-assisted charging in open-system dynamics~\cite{Cakmak2020, Francica2020, GhoshMal2021}, 
measurement-assisted extraction and protection of stored ergotropy~\cite{Malavazi2025, Chaki2026}, and many-body ergotropy fluctuations~\cite{Hovhannisyan2024}. 

These studies demonstrate that environmental interactions can provide viable mechanisms for increasing the extractable work of a quantum system.
However, they typically require resources that are far from minimal, such as external driving fields~\cite{Binder2015, Dou2022dicke} or repeated
interactions with streams of ancillas~\cite{Seah2021}. Others are constrained by asymptotic requirements: the optimal extractable
work per battery is reached only for asymptotically many copies~\cite{AlickiFannes2013}, and the free-energy bound on ergotropy extraction is saturated only for an infinite-dimensional working body~\cite{Biswas2022}. This raises a basic question: \emph{what is the minimal physical configuration in which an environment can inject ergotropy into a quantum system?}

This minimal-setting question is particularly relevant in the context of current quantum technologies. Qubits constitute the elementary building blocks of a wide range of quantum platforms, including superconducting circuits, trapped ions, and solid-state quantum devices~\cite{Leibfried2003, Doherty2013, Kjaergaard2020, Blais2021}. Thus, a qubit system interacting with a single qubit environment represents the smallest nontrivial system-environment configuration in which environmental charging can be investigated. Establishing whether ergotropy injection is possible in this setting, and determining the minimal modification required when it is not, provides a useful benchmark for the fundamental limits of environmental charging. 


A natural framework for addressing this question is provided by thermal operations (TO), the standard class of free operations in the resource theory of quantum thermodynamics~\cite{Janzing2000, Horodecki2013, Brandao2013, Brandao2015, Gour2015, Scharlau2018, Lostaglio2019}. In a thermal operation, the system is coupled to an environment initially prepared in a Gibbs state, and the joint system evolves under a unitary that conserves the sum of the bare system and environment Hamiltonians. {These constraints strongly restrict the dynamics of coherence~\cite{LostaglioPRX2015, Lostaglio2015, Cwiklinski2015} and imply free-energy bounds on the extractable work~\cite{Biswas2022}.} Although ergotropy injection is possible under thermal operations, the conditions under which it can occur in the smallest finite-dimensional setting remain largely unexplored.

Here we address this question by considering the minimal qubit-qubit configuration. We show that thermal operations can never inject ergotropy into the system in the qubit-qubit setting, irrespective of the environment temperature or energy gap. This protection is specific to the qubit environment: increasing the environment dimension by only one level is sufficient to enable ergotropy injection. In particular, we show that a qutrit environment can generate a positive ergotropic gain in a qubit system. Thus, within the thermal-operation framework, the qutrit is the minimal environmental dimension capable of injecting ergotropy into a qubit.

This result naturally raises the question of which assumption underlying thermal operations is responsible for the qubit-qubit no ergotropy gain result. Is it the requirement that the environment be thermal, or the restriction to energy conservation with respect to the sum of the bare Hamiltonians? To distinguish these possibilities, we go beyond thermal operations while retaining strict energy conservation. We allow an interaction term in the total Hamiltonian and consider an arbitrary, possibly athermal, initial state of the environment. The global unitaries are then required to commute with the \emph{full} Hamiltonian, including the interaction term, rather than only with the sum of the bare system and environment Hamiltonians. This naturally generalizes thermal operations by maintaining energy conservation while relaxing the assumptions of environmental thermality and a noninteracting total Hamiltonian. 

As a minimal realization of this generalized setting, we consider a qubit system coupled to a qubit environment through a isotropic XY exchange interaction~\cite{LSM1961, Jaynes1963}. The degeneracy structure of the conserved Hamiltonian determines the distinct regimes, the corresponding admissible energy-conserving unitaries, and, consequently, the mechanisms enabling ergotropy injection. In the central-block (CB) regime, energy-conserving unitaries act nontrivially within the single-excitation subspace, whereas in the double-degeneracy (DD) regime, the degeneracy extends between the zero- and two-excitation states and the single-excitation sector.

We first analyze the CB regime. We show that, unlike the thermal-operation setting, positive ergotropic gain can arise already in the minimal qubit--qubit configuration. For initially incoherent system and environment states, we find that the interaction by itself is insufficient to generate ergotropy, and that environmental athermality is essential for a positive gain. In fact we show that even in the absence of interactions, a population-inverted environment can generate the maximum possible ergotropic gain. In presence of interaction we obtain the corresponding optimal ergotropic gain and show that environmental coherence can further enhance it, while system coherence alone does not necessarily increase the gain and may instead be detrimental.

Furthermore, we demonstrate a positive ergotropic gain in the double-degeneracy regime and derive the optimal gain achievable in this regime. In contrast to the CB regime, we show that ergotropy injection is possible even when the environment is initially passive. In particular, positive ergotropic gain can be achieved even when the initial environment is more passive than the initial state of the system. {Closest to this setting, energy-conserving transport between two qubits has been shown to yield a net gain, but only when the two are initially correlated~\cite{Simon2025}; here no such correlation is required, the resource being the degeneracy structure of the conserved Hamiltonian.}

Our results therefore establish the minimal environmental dimension required for ergotropy injection under thermal operations and demonstrate how this limitation can be overcome in the minimal qubit-qubit setting by relaxing the assumptions of environmental thermality and a noninteracting total Hamiltonian. More broadly, they separate the roles of environment dimensionality, athermality, and system-environment interactions, providing a finite-dimensional setting in which the physical resources responsible for ergotropy injection can be identified. 

Figure~\ref{fig:schematic} provides a conceptual picture of the central idea of this work. The system $S$ is represented as attempting to pull an ``ergotropy box,'' symbolizing the acquisition of ergotropic gain by the system, while the environment $E$ can either remain unable to assist or actively aid this process, depending on the allowed operation and the resources it possesses. Under thermal operations, the environment is thermal and, despite interacting with the system, cannot assist in bringing the ergotropy box closer to $S$. In contrast, generalized energy-conserving operations  allow resources such as environmental athermality and system-environment interaction to contribute to the charging process. These resources effectively empower $E$ to assist $S$ in pulling the box closer, corresponding to ergotropic injection into the system.

The rest of the paper is organized as follows. In Sec.~\ref{2}, we briefly review the prerequisite concepts. In Sec.~\ref{3}, we present three theorems that together establish that no ergotropic gain is possible under thermal operations in the qubit–qubit setting. In Sec.~\ref{4}, we demonstrate, through an example, that increasing the bath dimension to a qutrit can activate a positive ergotropic gain. In Sec.~\ref{sec:model}, we introduce the XY-interaction model considered and classify the different types of admissible energy-conserving unitaries arising from the distinct degeneracy structures of the model. In Sec.~\ref{sec:ergotropy_interaction}, we begin our analysis of ergotropic injection in the classified regions. In Sec.~\ref{sec:CB}, we present the analysis of the CB regime, while in Sec.~\ref{sec:DD}, we investigate the double-degenerate (DD) regime. We present a summary of the results in Sec.~\ref{summary}.
Finally, we conclude in Sec.~\ref{sec:conclusion}. Appendix~\ref{app:DDsup} contains the proof of the optimal ergotropic gain in the DD regime.

\begin{figure}[t]
  \centering
  \makebox[\columnwidth][c]{
\includegraphics[width=0.47\textwidth,height=0.16\textwidth]
{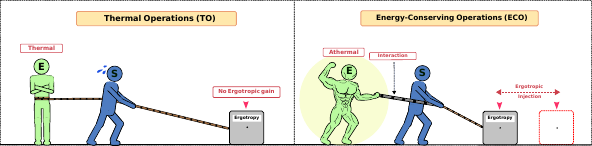}
}
  \caption{Conceptual schematic illustrating the central message of this work. The system $S$ is depicted as pulling an ``ergotropy box,'' representing ergotropic gain in the system. Under thermal operations (TO) (left), the environment $E$ is thermal and cannot assist $S$ in pulling the box closer, indicating the absence of ergotropic gain. Under generalized energy-conserving operations (ECO) (right), environmental athermality and system--environment interaction empower $E$ to assist $S$, thereby enabling ergotropic injection into the battery.}
  \label{fig:schematic}
\end{figure}

We begin by briefly introducing the prerequisites for our analysis, including ergotropy, ergotropic gain, and the definitions of the different types of energy-preserving operations considered in this work. These are defined as follows.

\section{Preliminaries}
\label{2}
\subsection{Ergotropy and Ergotropic Gain}
Let $\rho$ be the state of a quantum system with Hamiltonian $H_S^{d_S}$ acting on a $d_S$-dimensional Hilbert space $\mathcal{H}_S$. Let
$V\in\mathcal{U}(\mathcal{H}_S^{d_S})$ denote the unitary operator corresponding to a cyclic unitary process, where $\mathcal{U}(\mathcal{H}_S)$ is the set of all
unitary operators acting on $\mathcal{H}_S$. Then the ergotropy~\cite{Allahverdyan2004} is defined as the maximum amount of work that
can be extracted from the system through such a process and is given by
\begin{equation}
R(\rho)=\Tr(H_S\rho)-\min_{V\in\mathcal{U}(\mathcal{H}_S^{d_S})}
\Tr(H_SV\rho V^\dagger).
\label{eq:ergotropy}
\end{equation}

Let
$\rho=\sum_k r_k\ketbra{r_k}{r_k}$
and
$H_S=\sum_k\epsilon_k\ketbra{k}{k}$
be the spectral decompositions of the state and the Hamiltonian, respectively,
where the eigenvalues satisfy
$r_1\ge r_2\ge\cdots\ge r_k$
and
$\epsilon_1^S\le\epsilon_2^S\le\cdots\le\epsilon_k^S$. The minimum in Eq.~\eqref{eq:ergotropy} is attained by assigning the eigenvalues of
$\rho$ to the energy eigenstates in decreasing order of population and increasing order of energy, respectively~\cite{Pusz1978,Lenard1978}.
The resulting state,
$
\rho^{p}=\sum_k r_k\ketbra{k}{k},
$
is called the \emph{passive state}. It satisfies
$[\rho^{p},H_S]=0$ and possesses zero ergotropy. Consequently,
\begin{equation}
R(\rho)
=
\Tr(H_S\rho)-\sum_k r_k\epsilon_k^S
\equiv
\Tr(H_S\rho)-E_P(\rho),
\label{eq:ergotropy2}
\end{equation}
where
$E_P(\rho)=\sum_k r_k\epsilon_k^S$
denotes the passive energy.

In general, a quantum system is not isolated but interacts with its surrounding environment. Let the environment be described by a Hilbert space $\mathcal{H}_E^{d_E}$, an initial state $\rho_E$, and a Hamiltonian $H_E$, where $d_E$ denotes the dimension of the environment. The joint
system--environment evolution is governed by a unitary operator
$U_{SE}$ acting on
$\mathcal{H}_S\otimes\mathcal{H}_E$. The corresponding reduced evolution of the
system is described by the quantum channel
\begin{equation}
\Lambda(\rho)
=
\Tr_E\!\left[
U_{SE}
(\rho\otimes\rho_E)
U_{SE}^{\dagger}
\right],
\label{lam}
\end{equation}
where $\Tr_E$ denotes the partial trace over the environment.

Since the state of the system changes under $\Lambda$, its ergotropy generally
changes as well. For a given initial state $\rho$, we therefore define the
\emph{ergotropic gain} as
\begin{equation}
\Delta R
:=
R\!\left(\Lambda(\rho)\right)-R(\rho).
\label{eq:deltar}
\end{equation}
A positive value of $\Delta R$ indicates that the interaction has increased the
extractable work stored in the system, effectively charging it from the
viewpoint of ergotropy. Conversely, $\Delta R<0$ signifies a loss of
extractable work, corresponding to a discharge of the quantum battery.

In this work, we investigate quantum channels generated by
\emph{energy-conserving operations} (ECOs), encompassing both the standard
thermal operations (TOs) and a more general class of energy-preserving
dynamics. 

The different classes of energy-conserving operations considered in this work
are introduced below.
\medskip

\subsection{Energy-Conserving Operations}
\label{EC}
The total Hamiltonian of the system--environment composite, in the absence of any interaction between them, is given by
\begin{equation}
H_{SE}=H_S\otimes I_E+I_S\otimes H_E.
\label{eq:totalH}
\end{equation}
Now, consider the quantum channel $\Lambda(\rho)$, defined in Eq.~\ref{lam}. Such a channel is called a thermal operation (TO)~\cite{Janzing2000, Horodecki2013} if it satisfies the following conditions:
\begin{itemize}
\item The environment is initially prepared in a thermal state,
\begin{equation}
    \rho_E=\gamma_E=e^{-\beta H_E}/\Tr[e^{-\beta H_E}],
    \label{env}
\end{equation}
where $\beta$ is the inverse temperature of the environment.
\item The global unitary $U_{SE}$ conserves the total local energy, i.e.,
$[U_{SE},H_S\otimes I_E+I_S\otimes H_E]=0$.
\item Consequently, the thermal state of the system remains invariant under the operation,
$\Lambda(\gamma_S)=\gamma_S$, where
$\gamma_S=e^{-\beta H_S}/\Tr[e^{-\beta H_S}]$.
\end{itemize}

A more general class of energy-conserving quantum operations is obtained by relaxing the assumption that the environment is initially thermal and allowing an interaction between the system and the environment. In this case, the environment may be prepared in an arbitrary state, and the total Hamiltonian takes the form
\begin{equation}
H_{SE}=H_S\otimes I_E+I_S\otimes H_E+H_{\mathrm{int}},
\end{equation}
where $H_{\mathrm{int}}$ is the system--environment interaction Hamiltonian. The corresponding energy conservation condition becomes
$[U_{SE},H_{SE}]=0$. These operations strictly generalize thermal operations, since thermal operations correspond to the special case $H_{\mathrm{int}}=0$ together with a thermal environmental state.

In this work, we systematically investigate ergotropy injection in finite-dimensional systems. We first consider a qubit system interacting with a qubit environment and prove that thermal operations are incapable of increasing the system's ergotropy. We then extend our analysis to a qutrit system and uncover a striking dimensional transition: unlike the qubit case, thermal operations alone become sufficient to inject ergotropy into the system. Furthermore within the qubit-qubit setting we show that ergotropy injection becomes possible considering the broader class of energy-conserving operations. More remarkably, even in the absence of an interaction Hamiltonian ($H_{\mathrm{int}}=0$), merely preparing the environment in an athermal state is sufficient to increase the ergotropy of the system. This demonstrates that the non-thermal nature of the environment itself can act as a resource for charging a quantum system, without relying on direct system--environment interactions.


We present our findings for a qubit system interacting with a qubit environment in the following section.


\section{Impossibility of Ergotropy Injection under TO for a Qubit System and a Qubit Environment}
\label{3}

Before showing that generalized energy-conserving operations can inject
ergotropy, we first establish that this is impossible under standard
thermal operations (TOs). Throughout this section, we restrict our
attention to a qubit system, i.e., $d_S=2$. We prove this impossibility in two complementary scenarios. The first theorem addresses thermal operations in the non-degenerate regime, while the second and third consider thermal operations in the degenerate regime.

A thermal operation is said to be implemented in the {non-degenerate} regime if no Bohr frequency of the system coincides with any Bohr
frequency of the environment~\cite{Cwiklinski2015}. Equivalently, if the spectra
$\{\epsilon_i^S\}_{i=0}^{d_S-1}$ of $H_S$ and $\{\epsilon_k^E\}_{k=0}^{d_E-1}$ of $H_E$ satisfy
\begin{equation}
\epsilon_i^S+\epsilon_k^E=\epsilon_j^S+\epsilon_l^E
\quad\Longrightarrow\quad i=j,\;k=l,
\label{eq:nonres}
\end{equation}
then the system--environment Hamiltonian pair is said to satisfy the
non-resonance condition. Otherwise, if there exists at least one
quadruple $(i,j,k,l)$ with $i\neq j$ and $k\neq l$ satisfying
\begin{equation}
\epsilon_i^S+\epsilon_k^E=\epsilon_j^S+\epsilon_l^E,
\label{non}
\end{equation}
the system-environment Hamiltonian pair is said to satisfy the
degenerate condition.

Since we consider only a qubit system, we first derive a closed-form
expression for its ergotropy, which will serve as a key ingredient in
the proofs of the subsequent theorems.

\begin{lemma}
\label{l1}
Let the qubit Hamiltonian be $H_S=h\sigma_z, h>0$, where $\sigma_z$ is the usual Pauli matrix, with $|0\rangle$ and $|1\rangle$ denoted as the excited and ground states,
respectively. And let $\rho$, be an arbitrary single-qubit state of the form
\begin{equation}
\rho=
\begin{pmatrix}
1-p & a\\
a^{*} & p
\end{pmatrix},
\qquad
\Delta\equiv1-2p,
\label{eq:qubitparam}
\end{equation}
where $p\in[0,1]$ is the ground-state population, the ergotropy of
$\rho$ is given by
\begin{equation}
R(\rho)=h\,f(\Delta,|a|),
\label{eq:qubit_ergotropy}
\end{equation}
where
\begin{equation}
f(\Delta,r)
=
\Delta+\sqrt{\Delta^{2}+4r^{2}}
=
\frac{4r^{2}}
{\sqrt{\Delta^{2}+4r^{2}}-\Delta}.
\label{eq:fdef}
\end{equation}
\end{lemma}

\begin{proof}
The eigenvalues of $\rho$ are $$\lambda_{\pm}=\frac12\left(1\pm\sqrt{\Delta^2+4|a|^2}\right),$$
and its average energy is
$\Tr(H_S\rho)=h\Delta$.
The passive state associated with $\rho$ is
$\rho_P=\lambda_{+}\ketbra{1}{1}+\lambda_{-}\ketbra{0}{0}$,
whose average energy is
$\Tr(H_S\rho_P)=-h\sqrt{\Delta^2+4|a|^2}$.
Using the definition of ergotropy [Eq.~\eqref{eq:ergotropy2}], we obtain
\begin{align*}
R(\rho)
&=\Tr(H_S\rho)-\Tr(H_S\rho_P)\\
&=h\left(\Delta+\sqrt{\Delta^2+4|a|^2}\right)\\
&=h\,f(\Delta,|a|),
\end{align*}
which completes the proof.
\end{proof}

Next we present the theorem for no ergotropy injection by thermal operations, in the non-degenerate regime.

\begin{theorem}
\label{th1}
Let $\Phi$ be the thermal operation defined in Sec.~\ref{EC}, acting on a
qubit state $\rho$ of the form given in Eq.~\eqref{eq:qubitparam}.
Assume that the spectra of $H_S$ and $H_E$ satisfy the
non-resonance condition~\eqref{eq:nonres}. Then the change in
ergotropy, $\Delta R \equiv R(\Phi(\rho))-R(\rho)$, satisfies
$\Delta R\le0$.
\end{theorem}

\begin{proof}
Let the spectral decomposition of the environment Hamiltonian be
$H_E=\sum_{k=0}^{d_E-1}\epsilon_k^E\ketbra{\tilde{k}}{\tilde{k}}$, and let
the initial state of the environment be
$\rho_E=\sum_{k=0}^{d_E-1}t_k\ketbra{\tilde{k}}{\tilde{k}}$,
where $\{t_k\}$ is an arbitrary probability distribution. Consider an
energy-conserving unitary $U_{SE}$ acting on the joint state
$\rho\otimes\rho_E$, satisfying
$[U_{SE},H_S\otimes I_{d_E}+I_{d_S}\otimes H_E]=0$.

Since the spectra of $H_S$ and $H_E$ satisfy the non-resonance condition
[Eq.~\eqref{eq:nonres}], the total Hamiltonian
$H_T=H_S\otimes I_{d_E}+I_{d_S}\otimes H_E$
is nondegenerate. Consequently, every eigenspace of $H_T$ is
one-dimensional, and every energy-conserving unitary is necessarily
diagonal in the joint energy eigenbasis,
$
U_{SE}
=
\sum_{i,k}
e^{-i\theta_{ik}}
\ketbra{i,\tilde{k}}{i,\tilde{k}},
$
where $\theta_{ik}\in\mathbb{R}$.

The reduced dynamics of the system is therefore
\begin{equation}
\Lambda(\rho)
=
\Tr_E\!\left[
U_{SE}(\rho\otimes\rho_E)U_{SE}^{\dagger}
\right]
=
\sum_k
t_kV_k\rho V_k^{\dagger},
\end{equation}
where
$
V_k=\sum_i e^{-i\theta_{ik}}\ketbra{i}{i}
$
is a diagonal unitary acting on the system. Hence, $\Lambda$ is a random
unitary channel, i.e., the output state is a convex combination of
unitarily rotated copies of the input state.

Similarly, the reduced state of the environment is
$
\rho_E
=
\Tr_S\!\left[
U_{SE}(\rho\otimes\rho_E)U_{SE}^{\dagger}
\right]
=
\rho_E,
$
so the environment remains unchanged. Since both the system and the
environment preserve their average energies, the average energy of the
system is conserved,
$
\Tr[H_S\Lambda(\rho)]
=
\Tr[H_S\rho].
$

Therefore, the change in ergotropy is entirely determined by the passive
energies,
\begin{equation}
\Delta R
=
E(\rho_P)-E(\Lambda(\rho)_P),
\end{equation}
where $E(\sigma)=\Tr(H_S\sigma)$ and $\rho_P$ and
$\Lambda(\rho)_P$ denote the passive states corresponding to
$\rho$ and $\Lambda(\rho)$, respectively.

Furthermore, $\Lambda$ is unital since
$
\Lambda(I_S)=I_S.
$
By Uhlmann's theorem~{\cite{Bengtsson2006, NielsenChuang2010}}, every unital quantum channel is majorization-decreasing. Hence, if
$\mathbf{r}=(r_0,r_1)$ and
$\mathbf{r}'=(r_0',r_1')$
denote the eigenvalue vectors of $\rho$ and $\Lambda(\rho)$,
respectively, then
$
\mathbf{r}'\prec\mathbf{r}.
$
For a qubit this is equivalent to
$
r_0'\le r_0,
$
where the eigenvalues are arranged in non-increasing order.

Since the passive state is obtained by assigning the largest eigenvalue
to the ground state, we have
$
E(\rho_P)=r_0\epsilon_0^S+r_1\epsilon_1^S,
$
and
$
E(\Lambda(\rho)_P)
=
r_0'\epsilon_0^S+r_1'\epsilon_1^S.
$
Using $r_1=1-r_0$ and $r_1'=1-r_0'$, we obtain
\begin{equation}
\Delta R
=
(r_0-r_0')
(\epsilon_0^S-\epsilon_1^S).
\end{equation}
Since
$r_0-r_0'\ge0$
and
$\epsilon_0^S\le\epsilon_1^S$,
it follows that
$
\Delta R\le0.
$

The above proof holds for an arbitrary environment state diagonal in the
energy eigenbasis. Since the thermal state $\tau_E$ required in a
thermal operation is a particular instance of such a state, the result
applies directly to thermal operations $\phi(\rho)$. This completes the proof.
\end{proof}

Next we move on to the degenerate case.
In the following two theorems, we show that ergotopy injection even in the degenerate case is not possible if the system and environment are both qubits.

First, we consider the situation in which the degeneracy of the total Hamiltonian $H_T=H_S+H_E$ arises solely from degeneracies in the environment Hamiltonian $H_E$, while the system Hamiltonian $H_S$ is non-degenerate. In this case, we have
\begin{equation}
\epsilon_i^S+\epsilon_k^E=\epsilon_j^S+\epsilon_l^E
\quad\Longrightarrow\quad i=j.
\label{eq:nonres2}
\end{equation}
Throughout the paper, we refer to this as \emph{type-1 degeneracy}. If the above condition is satisfied without requiring $i=j$, we refer to it as \emph{type-2 degeneracy}.

The absence of ergotropy injection for type-1 degeneracy is established by the following theorem.

\begin{theorem}
Let $d_S=2$, and suppose that the spectra of $H_S$ and $H_E$ are such that the total Hamiltonian $H_T=H_S+H_E$ satisfies the type-1 degeneracy condition. Then no ergotropy can be injected into the system through a thermal operation, for an environment of arbitrary dimension $d_E$.
\end{theorem}

\begin{proof}
Let $\mathcal{S}$ denote the collection of degenerate sets of eigenstate
labels of $H_E$. We denote by $\mathcal{N}$ the set of eigenstate labels corresponding to non-degenerate eigenvalues of $H_E$. For a given degenerate energy $\epsilon_k^E$, let
$s_k\in\mathcal{S}$ denote the corresponding set of eigenstate labels, and
let $g_k$ denote its degeneracy, i.e., the number of eigenstates belonging
to $s_k$. We choose an orthonormal basis
\begin{equation}
\left\{
\ket{\tilde{k}^{(a)}}
\right\}_{a=1}^{g_k}
\end{equation}
for the subspace spanned by the eigenstates whose labels belong to $s_k$,
where we identify $\ket{\tilde{k}^{(1)}}\equiv\ket{\tilde{k}'}$, while the
remaining $g_k-1$ vectors are orthogonal to $\ket{\tilde{k}'}$. In
particular, if $\epsilon_k^E=\epsilon_m^E$, then $s_k=s_m$; hence, they
correspond to the same degenerate set and the same orthonormal basis.

The eigenstates $\{\ket{\tilde{k}}:\tilde{k}\in\mathcal{N}\}$, together
with the vectors $\{\ket{\tilde{k}^{(a)}}:s_k\in\mathcal{S},
\,a=1,\ldots,g_k\}$, form a complete orthonormal basis of
$\mathcal{H}_E$.

Since the total Hamiltonian satisfies the type-1 degeneracy condition,
energy conservation, $[U_{SE},H_T]=0$, implies that the global unitary can
only couple eigenvectors within the same total-energy eigenspace. Under
the type-1 condition, this means that the system energy remains unchanged,
and hence the unitary can be written as
\begin{equation}
\begin{split}
U_{SE}
&=
\sum_{i,\tilde{k}\in\mathcal{N}}
e^{-i\theta_{i\tilde{k}}}
\ketbra{i,\tilde{k}}{i,\tilde{k}}\\
&\quad+
\sum_i
\sum_{s_k\in\mathcal{S}}
\sum_{a=1}^{g_k}
e^{-i\theta_{i\tilde{k}^{(a)}}}
\ketbra{i,\tilde{k}^{(a)}}{i,\tilde{k}^{(a)}}.
\end{split}
\end{equation}

Thus, the reduced state of the system after the action of $U_{SE}$ on
$\rho\otimes\rho_E$, with
\begin{equation}
\rho_E=\sum_{k=0}^{d_E-1}t_{\tilde{k}}\ketbra{\tilde{k}}{\tilde{k}},
\end{equation}
takes the form
\begin{equation}
\begin{split}
\Lambda(\rho)
&=
\sum_{\tilde{k}\in\mathcal{N}}
t_{\tilde{k}}V_{\tilde{k}}\rho V_{\tilde{k}}^\dagger\\
&+
\sum_{s_k\in\mathcal{S}}
\sum_{\tilde{k}\in s_k}
\sum_{a=1}^{g_k}
t_{\tilde{k}}
\left|\braket{\tilde{k}^{(a)}|\tilde{k}}\right|^2
V_{\tilde{k}^{(a)}}\rho
V_{\tilde{k}^{(a)}}^\dagger,
\end{split}
\end{equation}
where
\begin{equation}
V_{\tilde{k}^{(a)}}=
\sum_i
e^{-i\theta_{i\tilde{k}^{(a)}}}
\ketbra{i}{i}.
\end{equation}

For a given $\tilde{k}\in s_k$, the completeness of the orthonormal basis
of the degenerate subspace gives
\begin{equation}
\sum_{a=1}^{g_k}
\left|\braket{\tilde{k}^{(a)}|\tilde{k}}\right|^2=1.
\end{equation}
Consequently,
\begin{equation}
\Lambda(I_S)=I_S,
\end{equation}
and hence the reduced map is unital.

We have already shown in Theorem~\ref{th1} that a unital quantum channel
acting on a qubit, for which the initial and final energies of the system
remain unchanged, cannot inject ergotropy into the system. Therefore, no
ergotropy injection is possible through thermal operations for type-1
degeneracy, irrespective of the environment dimension $d_E$.
\end{proof}
Next, we move on to the type 2 degenerate regime. In the following theorem, we show that even in this regime, considering a qubit system and qubit environment, ergotropy injection is not possible via thermal operation.

\begin{theorem}
    Let $d_S=d_E=2$, and the system-environment Hamiltonian pair satisfy the type 2 degenerate condition; then no ergotropy injection into the system is possible under TO, for any inverse temperature of the environment $\beta$.
\end{theorem}

\begin{proof}
Without loss of generality, let the system and environment Hamiltonians be $H_S=h_1\sigma_z, H_E=h_1\sigma_z,$ with $h_1>0$, where $\sigma_z$ is the usual Pauli operator. We denote its
eigenstates by $\{\ket{1},\ket{0}\}$, with corresponding eigenenergies
$\{-1,1\}$. The environment is initially prepared in the thermal state
\begin{equation}
    \tau_\beta=\operatorname{diag}(\tau_0,\tau_1),
    \qquad
    \tau_0=\frac{e^{-\beta h_1}}{2\cosh(\beta h_1)}
    \leq\frac{1}{2}\leq\tau_1,
\end{equation}
where $\tau_1=e^{\beta h_1}/[2\cosh(\beta h_1)]$.

Since
\begin{equation*}
    E_1^S+E_0^E=E_0^S+E_1^E=0,
\end{equation*}
the total Hamiltonian $H=H_S+H_E$ exhibits a twofold degeneracy in the
zero-energy subspace. Its spectrum is
\begin{equation*}
    \operatorname{spec}(H)=\{2h_1,0,0,-2h_1\}.
\end{equation*}
Consequently, an energy-conserving unitary $U$, satisfying $[U,H]=0$, must
act trivially, up to phases, on the nondegenerate energy eigenspaces, while
it can act arbitrarily within the degenerate zero-energy subspace. Thus,
\begin{equation}
    U=e^{\iota\phi_0}\ketbra{00}{00}
    \oplus V
    \oplus e^{\iota\phi_1}\ketbra{11}{11},
\end{equation}
where $V$ is an arbitrary $2\times2$ unitary acting on
$\operatorname{span}\{\ket{01},\ket{10}\}$. Unitarity implies
\begin{equation*}
    |V_{11}|=|V_{22}|=\sqrt{t},
    \qquad t\in[0,1].
\end{equation*}

Now consider an initial system state of the form
\begin{equation}
    \rho
    =
    p\ketbra{1}{1}
    +(1-p)\ketbra{0}{0}
    +a\ketbra{0}{1}
    +a^*\ketbra{1}{0}.
    \label{in}
\end{equation}
The reduced state of the system after the thermal operation generated by
$U$ is
\begin{equation*}
    \sigma
    =
    \Lambda(\rho)
    =
    \operatorname{Tr}_E
    \left[
        U(\rho\otimes\tau_\beta)U^\dagger
    \right].
\end{equation*}
It retains the same general form,
\begin{equation}
    \sigma
    =
    p'\ketbra{1}{1}
    +(1-p')\ketbra{0}{0}
    +a'\ketbra{0}{1}
    +a'^*\ketbra{1}{0},
    \label{fin}
\end{equation}
where
\begin{equation}
    p'
    =
    \tau_1+t\left(p-\tau_1\right),
    \qquad
    a'
    =
    a\left(
        \tau_0 e^{\iota\phi_0}V_{22}^*
        +\tau_1 V_{11}e^{-i\phi_1}
    \right).
    \label{eq:qubitenvironmentmaps}
\end{equation}
hence 

Using Lemma~\ref{l1}, the change in ergotropy induced by the thermal operation can be expressed as
\begin{equation}
\Delta R= R(\sigma)-R(\rho)=f(\Delta',|a'|)-f(\Delta,|a|)),
\end{equation}
where $\Delta'=1-2p'$. Substituting the explicit form of $f$, we obtain
\begin{equation}
\Delta R
=
h_1\Big(\Delta'-\Delta
+\sqrt{\Delta'^2+4|a'|^2}
-\sqrt{\Delta^2+4|a|^2}\Big).
\label{eq:ergotropy_change}
\end{equation}

We define $s=\Delta_\tau-\Delta$. It then follows from the expression for $p'$ that $
\Delta'=\Delta_\tau-ts$.
Furthermore, Eq.~\eqref{eq:qubitenvironmentmaps} and the triangle inequality give $|a'|\leq \sqrt{t}|a|$.

We first consider three special cases separately: (i) $t=0$, (ii) $t=1$, and (iii) $|a|=0$.

For $t=0$, we have $\Delta'=\Delta_\tau$. Since the initial environment state is passive, $\Delta_\tau\leq0$, and hence
\begin{equation}
R(\sigma)=h_1f(\Delta_\tau,0)=0.
\end{equation}
Thus, the final state is passive, and there is no gain in ergotropy.

For $t=1$, we have $\Delta'=\Delta$ and $|a'|\leq|a|$. Since $f(\Delta,r)$ is non-decreasing in $r\geq0$, it follows that
\begin{equation}
f(\Delta',|a'|)
=f(\Delta,|a'|)
\leq f(\Delta,|a|),
\end{equation}
and consequently $\Delta R\leq0$.

We next consider the case $|a|=0$. Then $|a'|=0$, and hence
\begin{equation}
\Delta R
=h_1\left[f(\Delta',0)-f(\Delta,0)\right].
\end{equation}
Using
\begin{equation}
\Delta'=1-2\left[\tau_1+t(p-\tau_1)\right],
\end{equation}
we consider two cases. If $p\geq\tau_1$, then, since $\tau_1\geq1/2$, we have $\Delta'\leq0$, and therefore $R(\sigma)=0$, implying $\Delta R\leq0$.

If $p\leq\tau_1$, then
\begin{equation}
\Delta'-\Delta
=2(1-t)(p-\tau_1)\leq0,
\end{equation}
so that $\Delta'\leq\Delta$. If $\Delta\leq0$, then $\Delta'\leq0$ and consequently $R(\sigma)=0$. If instead $\Delta>0$, then $\Delta'<\Delta$. Again, for this case, if $\Delta'<0$, the final state is passive and $R(\sigma)=0$. If $\Delta'\geq0$, then
\begin{equation}
\Delta R=2h_1(\Delta'-\Delta)<0.
\end{equation}
Thus, no ergotropy injection is possible when $|a|=0$.

We now consider the remaining regime, with $|a|\neq0$ and $0<t<1$. Define
\begin{equation}
g(t)=f\left(\Delta_\tau-ts,\sqrt{t}|a|\right),
\qquad t\in[0,1].
\end{equation}
By the triangle inequality,
\begin{equation}
f(\Delta',|a'|)
\leq
f\left(\Delta',\sqrt{t}|a|\right)
=g(t).
\end{equation}
Moreover, since $\Delta'(1)=\Delta$, we have
\begin{equation}
g(1)=f(\Delta,|a|).
\end{equation}
The function $g(t)$ is continuous on $[0,1]$. Therefore, it is sufficient to show that $g(t)$ is non-decreasing on $(0,1)$, in which case its maximum on $[0,1]$ is attained at $t=1$.

Differentiating $g(t)$ gives
\begin{equation}
\partial_tg(t)
=
\frac{
2|a|^2\left(A_t-\Delta'-2ts\right)
}{
A_t\left(A_t-\Delta'\right)
},
\label{eq:gt_derivative}
\end{equation}
where
\begin{equation}
A_t=\sqrt{\Delta'^2+4t|a|^2}.
\end{equation}
For $|a|\neq0$ and $t>0$, we have $A_t>|\Delta'|$, and hence $A_t-\Delta'>0$. Thus, Eq.~\eqref{eq:gt_derivative} is well defined in the present regime.

For completeness, the only exceptional points at which the denominator in Eq.~\eqref{eq:gt_derivative} can vanish satisfy $A_t=\Delta'\geq0$. This can occur only if $t|a|^2=0$. The cases $t=0$ and $|a|=0$ have already been treated above. 

We may therefore restrict attention to the generic case in which $|a|\neq0$, $0<t<1$. For this case since $A_t>-\Delta'$, we obtain
\begin{equation}
A_t-\Delta'-2ts>-2\Delta_\tau.
\end{equation}
For a thermal state, $\Delta_\tau\leq0$, and consequently
\begin{equation}
A_t-\Delta'-2ts>0.
\end{equation}
Together with $|a|^2>0$ and $A_t(A_t-\Delta')>0$, this yields
\begin{equation}
\partial_tg(t)>0,
\qquad t\in(0,1).
\end{equation}
Hence, $g(t)$ is strictly increasing on $(0,1)$. Since $g(t)$ is continuous on $[0,1]$, its maximum is attained uniquely at $t=1$, giving
\begin{equation}
f\left(\Delta',\sqrt{t},|a|\right)
\leq
f(\Delta,|a|).
\end{equation}
Combining this with
\begin{equation}
f(\Delta',|a'|)
\leq
f\left(\Delta',\sqrt{t},|a|\right),
\end{equation}
we obtain
\begin{equation}
\Delta R\leq0.
\end{equation}
Thus, no ergotropy can be injected in the generic regime either.

\end{proof}

\begin{remark}
The proof of $\Delta R\leq0$ in the above theorem relies on the passivity of the environment state. If the environment is not passive, the sign of $\Delta R$ cannot, in general, be determined from the above theorem. This suggests that an athermal environment is necessary for ergotropy injection in the absence of an interaction term in the Hamiltonian. In the subsequent section, we prove that this is indeed the case.
\end{remark}

Before exploring the role of athermality, we first show that, unlike Theorem 1 and theorem 2, "no egotropy injection" into the system does not hold for arbitrary environment dimension. In fact, considering a qutrit environment, we can inject ergotropy into the qubit system. The analysis of which is presented in the section below.

\section{Injection of ergotropy considering qutrit environment}
\label{4}
In the following example, we show that, in the presence of type-2 degeneracy, increasing the environment dimension to $d_E=3$ is sufficient to enable ergotropy injection into a qubit system through a thermal operation.

\begin{example}[Ergotropy injection from a qutrit environment]
\label{ex:qutrit}
Consider a qutrit environment with Hamiltonian
\begin{equation}
    H_E
    =
    \sum_{n=0}^{2}2h(n-1)
    \ketbra{\tilde{n}}{\tilde{n}},
\end{equation}
and a qubit system with Hamiltonian $H_S=h\sigma_Z$.

The total Hamiltonian $H=H_S+H_E$ exhibits type-2 degeneracy, since
\begin{equation}
    E_1^S+E_1^E=E_0^S+E_0^E=-h,
    \qquad
    E_1^S+E_2^E=E_0^S+E_1^E=h.
\end{equation}
Thus, the total Hilbert space contains two doubly degenerate energy
subspaces, in addition to two nondegenerate energy levels.

An energy-conserving unitary must act trivially, up to a phase, on the
nondegenerate subspaces, while it can act arbitrarily within each
degenerate subspace. We consider the unitary
\begin{equation}
    U
    =
    e^{\iota\phi_0}\ketbra{1\tilde{0}}{1\tilde{0}}
    \oplus U_1
    \oplus U_2
    \oplus
    e^{\iota\phi_1}\ketbra{0\tilde{2}}{0\tilde{2}},
\end{equation}
where $U_1$ and $U_2$ are swap unitaries acting, respectively, on the
degenerate subspaces
$\operatorname{span}\{\ket{1\tilde{1}},\ket{0\tilde{0}}\}$ and
$\operatorname{span}\{\ket{1\tilde{2}},\ket{0\tilde{1}}\}$. In particular,
\begin{equation}
    \ket{1\tilde{1}}\longleftrightarrow\ket{0\tilde{0}},
    \qquad
    \ket{1\tilde{2}}\longleftrightarrow\ket{0\tilde{1}}.
\end{equation}

Let the initial state of the system be the ground state
$\rho=\ketbra{1}{1}$, which is passive, and let the environment be
initially prepared in the thermal state $\gamma_E=
    \sum_{n=0}^{2}
    \tau_n\ketbra{\tilde{n}}{\tilde{n}}$,
with respect to $H_E$ at inverse temperature $\beta$. The corresponding
thermal populations are
\begin{equation}
    \tau_n
    =
    \frac{x^n}{1+x+x^2},
    \qquad
    x=e^{-2\beta h}.
\end{equation}

After the action of $U$ and tracing out the environment, the reduced state
of the system is
\begin{equation}
    \sigma
    =
    (\tau_1+\tau_2)\ketbra{0}{0}
    +\tau_0\ketbra{1}{1}.
\end{equation}
Hence, $\sigma$ becomes active whenever
$\tau_1+\tau_2>\tau_0$. In terms of $x$, this condition is equivalent to
\begin{equation}
    x+x^2-1>0.
\end{equation}
Since the energy gap of the system is $2h$, the resulting ergotropy is
\begin{equation}
    \Delta R
    =
    2h\max\left(
        0,\frac{x+x^2-1}{1+x+x^2}
    \right).
    \label{eq:qutritgain}
\end{equation}
Therefore, ergotropy injection occurs whenever
\begin{equation}
    x>\frac{\sqrt{5}-1}{2}
    =\frac{1}{\varphi},
\end{equation}
or, equivalently,
\begin{equation}
    2\beta h<\ln\varphi\approx0.481,
\end{equation}
where $\varphi=(1+\sqrt{5})/2$ is the golden ratio. Thus, in contrast to
the qubit-environment case, a three-dimensional thermal environment is
already sufficient to inject ergotropy into a qubit system under a
thermal operation.
\end{example}
As discussed in the preliminaries, TOs constitute a subclass of generalized energy-conserving operations, characterized by the absence of system-environment interaction terms in the total Hamiltonian and by an environment initially prepared in a thermal state. In Theorems~1-3, we considered the non-interacting setting while allowing the initial environment state to be an arbitrary diagonal state; only Theorem~3 required the additional assumption of passivity to establish the impossibility of ergotropy injection. This suggests two natural routes for moving beyond the restrictions underlying TOs: introducing system-environment interactions and allowing the environment to possess athermality. In the following, we investigate how these two resources affect the possibility of ergotropy injection into a qubit system. To this end, we consider the following Hamiltonian model.

\section{Two-qubit model with exchange coupling}
\label{sec:model}
Let $d_S=d_E=2$, and consider the total Hamiltonian of the system-environment pair to be
\begin{equation}
\begin{split}
H &= h_1\sigma_z\otimes\mathbb{I}+h_2\mathbb{I}\otimes \sigma_z+H_{SE},\\
H_{SE} &= J\left(\sigma_+\otimes\sigma_-+\sigma_-\otimes\sigma_+\right),
\end{split}
\label{eq:JCH}
\end{equation}
where $h_1,h_2>0$ and, without loss of generality, $J\geq0$. The exchange interaction $H_{SE}$ is the two-qubit isotropic XY coupling~\cite{LSM1961}, which exactly conserves the total excitation number. The same excitation-conserving structure underlies the Jaynes--Cummings model~\cite{Jaynes1963}, where the second party is a bosonic mode rather than a qubit.

The spectrum of $H$ is
\begin{equation}
\{h_1+h_2,E_c,-E_c,-(h_1+h_2)\},
\end{equation}
with the corresponding eigenstates
\begin{equation}
\{\ket{00},\ket{+},\ket{-},\ket{11}\},
\end{equation}
where the dressed states are given by
\begin{equation}
\label{eq:dressed}
\begin{split}
\ket{+} &= \cos\theta_c\ket{01}+\sin\theta_c\ket{10},\\
\ket{-} &=-\sin\theta_c\ket{01}+\cos\theta_c\ket{10},
\end{split}
\end{equation}
with
\begin{equation}
E_c=\sqrt{\kappa^2+J^2},
\sin 2\theta_c=\frac{J}{E_c},
\cos 2\theta_c=\frac{\kappa}{E_c},
\kappa=h_1-h_2,
\label{eq:thetac}
\end{equation}
and $\theta_c\in[0,\pi/2]$.

Depending on the values of $h_1$, $h_2$, and $J$, the spectrum exhibits two distinct types of degeneracy:
\begin{enumerate}
\item[(a)] \emph{Central-block (CB) degeneracy :} This occurs when $E_c=0$, which is possible if and only if
\begin{equation}
J=0,\qquad h_1=h_2.
\end{equation}

\item[(b)] \emph{Cross-sector degeneracy:} This occurs when
\begin{equation}
E_c=h_1+h_2,
\end{equation}
which is equivalent to
\begin{equation}
J=J_c\equiv2\sqrt{h_1h_2}.
\label{eq:Jc}
\end{equation}
At this point, the dressed state $\ket{+}$ becomes degenerate with $\ket{00}$, while $\ket{-}$ becomes degenerate with $\ket{11}$. Furthermore,
\begin{equation}
\sin2\theta_c
=\frac{J_c}{h_1+h_2}
=\frac{2\sqrt{h_1h_2}}{h_1+h_2}.
\end{equation}
\end{enumerate}
No other types of degeneracy occur.

We now consider generalized energy-preserving operations on the two-qubit system, having the total Hamiltonian $H$. We analyze the situation when ergotropy can be injected into the system using such kind of operations. If the global unitary corresponding to the operation is $U$, then the admissible unitaries satisfying $[U,H]=0$ are completely determined by the degeneracy structure. Below we classify such admissible unitaries.

\subsection{Classification of admissible unitaries}
\label{subsec:classification}

For the Hamiltonian~\eqref{eq:JCH}, the admissible unitaries are those
satisfying $[U, H]=0$. Their form is determined by the degeneracy structure
of the spectrum of $H$. In the basis $\{\ket{00},\ket{01},\ket{10},\ket{11}\}$, the different cases are classified as follows.

\begin{enumerate}
    \item[(i)] \textbf{Non-interacting case} ($J=0$). This case can be
    further divided into two subcases:
    
    \begin{enumerate}
        \item $h_1\neq h_2$: The spectrum is nondegenerate, and hence the
        admissible unitary is diagonal in the computational basis,
        \begin{equation}
            U=
            \operatorname{diag}\left(
                e^{\iota\phi_{00}},
                e^{\iota\phi_{01}},
                e^{\iota\phi_{10}},
                e^{\iota\phi_{11}}
            \right),
         \label{eq:U_nonint}
        \end{equation}
        
        \item The central block degeneracy, with $h_1=h_2$: The two single-excitation states
        $\ket{01}$ and $\ket{10}$ are degenerate, corresponding to the
        type-2 degeneracy discussed in Sec.~\ref{3}. Consequently,
        the admissible unitary is diagonal on the nondegenerate
        eigenspaces and arbitrary within the degenerate subspace,
        \begin{equation}
            U=
            e^{\iota\phi_0}\ketbra{00}{00}
            \oplus V
            \oplus
            e^{\iota\phi_1}\ketbra{11}{11},
            \label{CBgenu}
        \end{equation}
        where $V\in U(2)$ acts on
        $\operatorname{span}\{\ket{01},\ket{10}\}$.
        
    \end{enumerate}

    \item[(ii)] \textbf{Interacting case} ($J>0$).
    \begin{enumerate}
      \item The central block non-degenerate regime with $J\neq J_c$:  The spectrum is nondegenerate.  Hence, the admissible unitary is diagonal in the dressed-state
        basis,
        \begin{equation}
            U=
            e^{\iota\phi_0}\ketbra{00}{00}
            \oplus V
            \oplus
            e^{\iota\phi_1}\ketbra{11}{11},
         \label{eq:U_CB}
        \end{equation}
        where
        \begin{equation}
            V=
            e^{\iota\phi_+}\ketbra{+}{+}
            +
            e^{\iota\phi_-}\ketbra{-}{-}.
            \label{V}
        \end{equation}
        
        \item \textbf{Cross-sector degeneracy} ($J=J_c$): At $J_c=2\sqrt{h_1h_2}$, the spectrum exhibits two independent
        doubly degenerate energy sectors. The admissible unitary therefore
        takes the form
        \begin{equation}
            U=W_+\oplus W_-,
            \label{eq:U_DD}
        \end{equation}
        where $W_+\in U(2)$ acts on
        $\operatorname{span}\{\ket{00},\ket{+}\}$ and
        $W_-\in U(2)$ acts on
        $\operatorname{span}\{\ket{-},\ket{11}\}$. The dressed-state
        mixing angle at the degeneracy point satisfies
        \begin{equation}
            \sin 2\theta_c
            =
            \frac{J_c}{h_1+h_2},
            \qquad
            \cos 2\theta_c
            =
            \frac{h_1-h_2}{h_1+h_2}.
        \end{equation}
  \end{enumerate}
\end{enumerate}
Note that $\phi_{ij}\in \mathbb{R}$. Next, we analyze these different regimes and see when ergotropy injection into the system is possible.
\section{Analysis of Ergotropy Injection in the Defined Regimes}
\label{sec:ergotropy_interaction}

We consider the initial state of the system $\rho$ to be of the form given in Eq.~\eqref{in}, while the initial state of the environment is taken to be
\begin{equation}
    \rho_E=\sum_{m,n}\tau_{mn}\ketbra{m}{n}.
    \label{env2}
\end{equation}
Unlike the case of thermal operations (TO), the initial state of the environment is not required to be thermal. We therefore refer to such generalized states of the environment as athermal states. The reduced state of the system after the energy-conserving operation is given by
\begin{equation}
    \sigma=\Tr_E[U(\rho\otimes\rho_E)U^\dagger],
\end{equation}
has the form as given in Eq.~\eqref{fin}. We first analyze ergotropy injection in the non-interacting regime.

\subsection{The non-interacting case with \texorpdfstring{$(h_1 \neq h_2)$}{h1 != h2}}

\textbf{\textit{No ergotropy injection is possible for case 1(a):}}
\label{cor:nonint}

For $U$ of the form given in Eq.~\eqref{eq:U_nonint}, the diagonal elements of $\sigma$ remain unchanged. Consequently, $p'=p$ and $\Delta'=\Delta$. On the other hand, the coherence transforms as
$a'=Ka$, where
\begin{equation}
    K=\sum_{m=0}^1\tau_{mm}\,e^{\iota(\phi_{0m}-\phi_{1m})}.
\end{equation}
Using the triangle inequality, we have $|K|\leq 1$. Therefore,
\begin{equation}
    \Delta R=h_1\big(f(\Delta,|K||a|)-f(\Delta,|a|)\bigr)\leq 0.
\end{equation}
Thus, the ergotropy cannot increase under the energy-conserving operation in this case. 

In the following section, we analyze the CB regime, which comprises both cases 1(b) and 2(a).

\section{Ergotropic Injection in CB regime}
\label{sec:CB}
With $U$ of the form~\eqref{CBgenu}, a direct computation of the elements of the final state yields
\begin{align}
  p'&= \tau_{11}+V_{22}^2(p-\tau_{11})
  +2\operatorname{Re}\!\left[V_{11}a\tau_{01}^{*}V_{21}^{*}\right],
  \label{eq:master_pop}\\[2pt]
  a'&=\lambda_1\,a+\lambda_2\,\tau_{01},
  \label{eq:master_coh}
\end{align}
where the \emph{retention} and \emph{conversion} amplitudes are given by
\begin{align}
  \lambda_1 &=
  (1-\tau_{11})e^{i\phi_0}V_{22}^{*}
  +\tau_{11}e^{-i\phi_1}V_{11},
  \nonumber\\
  \lambda_2 &=
  (1-p)e^{i\phi_0}V_{21}^{*}
  +p\,e^{-i\phi_1}V_{12}.
  \label{eq:lambdas}
\end{align}
Here, $V_{km}$, with $k,m\in\{1,2\}$, denotes the $(k,m)$-th matrix element of the
$2\times2$ unitary matrix $V$, i.e.,
\begin{equation}
V=
\begin{pmatrix}
V_{11} & V_{12}\\
V_{21} & V_{22}
\end{pmatrix}.
\end{equation}

From the orthogonality condition of the CB states, we have $|V_{11}|=|V_{22}|$, $|V_{12}|=|V_{21}|$.

We next consider different classes of initial system-environment states and
analyze ergotropy injection into the system.

\subsection{\textbf{Initially Incoherent System and Environment}}
\label{subsec:CBincoherent}

Note that, for $a=\tau_{01}=0$, the output state of the system takes a form similar to that obtained in Theorem~3. However, unlike in Theorem~3, the initial state of the environment is not required to be thermal, which was an essential ingredient in the proof of that theorem.

\textit{The CB-degenerate case:}
We first consider the CB-degenerate case. Here, the injected ergotropy can attain its maximum possible value. Let the initial states of the system and environment be $\rho=\ketbra{1}{1}$ and $\rho_E=\ketbra{0}{0}$, respectively. Since $V$ can be chosen as an arbitrary unitary in $U(2)$, we may choose it to implement a swap between $\ket{01}$ and $\ket{10}$. Consequently, the total unitary $U$ acts as a swap on the relevant subspace, transforming the initial system state $\ket{1}$ into $\ket{0}$. The resulting final state of the system is therefore $\ket{0}\!\bra{0}$, and the corresponding injected ergotropy is
\begin{equation}
    \Delta R=R'=2h_1.
\end{equation}

In contrast, for the CB-nondegenerate case, $V$ cannot be chosen arbitrarily, as its form is constrained by the Hamiltonian parameters $h_1$, $h_2$, and $J$. We therefore determine, in the following proposition, the maximum ergotropy injection attainable under these constraints.
\begin{proposition}
For the CB-nondegenerate case, the maximum ergotropy injection, optimized over all choices of the initial states $\rho$ and $\rho_E$ and over all energy-conserving unitaries $U$, is given by
\begin{equation}
    \Delta R_{\max}
    =2h_1\sin^2 2\theta_c
    =\frac{2h_1J^2}{J^2+\kappa^2}.
\end{equation}
\label{prop1}
\end{proposition}

\begin{proof}
The injected ergotropy is given by
\begin{equation*}
    \Delta R=R(\sigma)-R(\rho).
\end{equation*}
The first term depends on the choice of $\rho_E$ and the energy-conserving unitary $U$, whereas the second term depends only on $\rho$. Since
\begin{equation}
    \Delta R \leq R(\sigma)-\min_{\rho} R(\rho).
\end{equation}
we can choose $\rho$ to be a passive state, for which $R(\rho)=0$. Thus, maximizing $\Delta R$ is equivalent to maximizing $R(\sigma)$ over all choices of $U$, $\rho_E$, and passive initial states $\rho$.

Using the expression for $R(\sigma)$, and putting $R(\rho)=0$ we obtain
\begin{equation}
    \Delta R \leq h_1\left(\Delta'+|\Delta'|\right).
\end{equation}
Clearly, the upper bound is maximized for choices of parameters for which $\Delta'>0$. Hence,
\begin{equation}
    \Delta R \leq 2h_1(1-2p'),
\end{equation}
where $p'$ is given by Eq.~\eqref{eq:master_pop} with $\tau_{01},a=0$. Putting this in, we get,
\begin{equation}
    \Delta R \leq 2h_1\left[1-2\left(\tau_{11}+|V_{22}|^2(p-\tau_{11})\right)\right].
\end{equation}
In order to further optimize the bound we need to minimize
\begin{equation}
    \tau_{11}+|V_{22}|^2(p-\tau_{11}).
\end{equation}
From the expression for $V$ in Eq.~\eqref{V}, we have
\begin{equation}
    |V_{22}|^2
    =\cos^4\theta_c+\sin^4\theta_c
    +2\cos^2\theta_c\sin^2\theta_c
    \cos(\phi_+-\phi_-).
\end{equation}

Consider first the case $p<\tau_{11}$. In this case, minimizing the expression above requires maximizing $|V_{22}|^2$, which is attained for $|V_{22}|^2=1$. However, this choice yields $\Delta' \leq (1-2p)$, and since $p \geq \frac{1}{2}$, we get $\Delta R\leq0$. We therefore consider the case $p>\tau_{11}$. In this regime, minimizing the expression requires minimizing $|V_{22}|^2$, which is achieved for $\phi_+=n\pi+\phi_-$, where $n$ is any odd integer. The corresponding minimum value is $|V_{22}|^2=\cos^2{2\theta_c}$.
This gives
\begin{equation}
    \Delta R
    \leq 2h_1\Big[1-2\left(\sin^2 2\theta_c\,\tau_{11}
    +\cos^2 2\theta_c\,p\right)\Big].
\end{equation}
To further optimize over the initial states, we choose $\tau_{11}=0$. Since for the bound to be strict $\rho$ must be passive, $p$ cannot be zero; its minimum allowed value is $p=1/2$. Hence, the maximum ergotropy injection is
\begin{equation}
    \Delta R_{\max}
    =2h_1\sin^2 2\theta_c.
\end{equation}
\end{proof}

\begin{remark}
   \textbf{\textit{No ergotropy injection with initially passive environment states:}}
In Theorem~3, we have already shown that, in the CB-degenerate case, no ergotropy can be injected when the initial state of the {environment?} is thermal. For the CB-nondegenerate case, when the initial states are restricted to be passive, the structures of the initial and final states become identical to those considered in Theorem~3. The only difference is that, in the present case, $|V_{22}|^2$ (denoted by $t$ in Theorem~3) is fixed by the value of $\theta_c$. Nevertheless, the proof of Theorem~3 holds for arbitrary $t\in[0,1]$. Therefore, the same argument applies here, and we conclude that no ergotropy injection is possible in the CB-nondegenerate case when the initial states are passive.
\end{remark}

\begin{remark}
   From Remark~2, it follows that, for initially incoherent states and considering only the CB-nondegenerate case with finite interaction strength $J \neq 0$, the presence of interaction alone is insufficient to induce ergotropy injection when the environment is initially in a thermal state. Therefore, in this setting, population inversion, or more generally, athermality of the environment state, constitutes a necessary resource for ergotropy injection.
\end{remark}
Next we consider the situation when the initial state of the system poses some coherence, but the environment initial state is still incoherent.

\subsection{Initial States with  System Coherence Alone}
\begin{proposition}
\label{SCCB}
If the system has coherence $a\neq0$ while the environment is incoherent, $b=0$, but athermal, then, for a given pair of initial states $\rho,\rho_E$ and a fixed unitary of the form given in Eq.~\eqref{CBgenu}, the injected ergotropy satisfies
\end{proposition}
\begin{equation*}
\Delta R\leq 2h_1\max[\Delta'-\Delta,0].
\end{equation*}
\label{prop2}
\label{subsec:CBsyscoh}

\begin{proof}
For a given pair of initial states $\rho,\rho_E$ and a fixed unitary of the form given in Eq.~\eqref{CBgenu}, the ergotropic gain can be written as
\begin{equation*}
\Delta R=f(\Delta',|a'|)-f(\Delta,|a|).
\end{equation*}
For an incoherent environment, $b=0$, one has
$\Delta'-\Delta=2|V_{21}|^{2}(p-\tau_{11})$, which is independent of the system coherence $a$. Moreover,
$|a'|=|\lambda_1||a|\leq |a|$. Since $f(\Delta',|a|)$ is monotonically increasing with respect to $|a|$, it follows that
\begin{equation}
\Delta R\leq f(\Delta',|a|)-f(\Delta,|a|).
\end{equation}

For $a=0$, the desired bound follows directly from the triangle inequality,
\begin{equation}
|\Delta'|-|\Delta|\leq |\Delta'-\Delta|.
\end{equation}
For $a\neq0$, consider the function $f(x,|a|)$ as a function of $x$. By the mean value theorem, there exists a point $x$ lying between $\Delta$ and $\Delta'$ such that
\begin{equation*}
\left|\sqrt{\Delta'^{2}+4|a|^{2}}-\sqrt{\Delta^{2}+4|a|^{2}}\right|
=\left|f'(x,|a|)\right||\Delta'-\Delta|.
\end{equation*}
Since
\begin{equation*}
|f'(x,|a|)|=\frac{|x|}{\sqrt{x^{2}+4|a|^{2}}}\leq1,
\end{equation*}
we obtain
\begin{equation*}
\left|\sqrt{\Delta'^{2}+4|a|^{2}}-\sqrt{\Delta^{2}+4|a|^{2}}\right|
\leq|\Delta'-\Delta|.
\end{equation*}
Consequently,
\begin{equation*}
\Delta R\leq h_1\left[(\Delta'-\Delta)+|\Delta'-\Delta|\right]
=2h_1\max[\Delta'-\Delta,0].
\end{equation*}
This establishes the stated upper bound.

Since
\begin{equation*}
\Delta'-\Delta=2|V_{21}|^{2}(p-\tau_{11}),
\end{equation*}
the ergotropic gain can be positive only when $p>\tau_{11}$; for $p\leq\tau_{11}$, one has $\Delta'\leq\Delta$ and hence
\begin{equation*}
\Delta R\leq0.
\end{equation*}
Thus, in this regime, coherence in the system cannot induce a positive ergotropic injection.

More importantly, when $p>\tau_{11}$, the bound becomes
\begin{equation*}
\Delta R\leq2h_1(\Delta'-\Delta)
=4h_1|V_{21}|^{2}(p-\tau_{11}).
\end{equation*}
For $a=0$, the triangle inequality can be saturated when  $(\Delta'-\Delta)\Delta \geq 0$. Hence, for the same initial environment state and the same unitary, the incoherent system can attain the upper bound. In contrast, for $a\neq0$, one has
\begin{equation}
|f'(x,|a|)|=\frac{|x|}{\sqrt{x^{2}+4|a|^{2}}}<1,
\label{inprop2}
\end{equation}
for finite $x$. Therefore, whenever $\Delta'\neq\Delta$, the inequality arising from the mean value theorem is strict. In addition, $|a'|\leq|a|$ can only further decrease the ergotropic gain. Hence, for the same initial populations, the same environment state, and the same unitary, introducing coherence into the system cannot enhance the ergotropic injection and, whenever the incoherent case yields a positive gain, generally makes the gain strictly smaller.
This completes the proof.

\begin{remark}
For the same initial populations, environment state, and unitary, the present argument establishes that coherence cannot enhance the ergotropic injection whenever the incoherent case saturates the triangle inequality, i.e., $(\Delta'-\Delta)\Delta \geq0$. In this regime, the incoherent state attains the upper bound(An example is the optimal case in Proposition~\ref {prop1}), whereas for $a\neq0$ the mean-value-theorem inequality is strict (for $\Delta'\neq\Delta$), and $|a'|\leq |a|$, so that the presence of coherence can only decrease the ergotropic injection. When $(\Delta'-\Delta)\Delta <0$, however, the triangle inequality is not saturated for the incoherent state, and the above argument does not suffice to establish whether coherence is advantageous or disadvantageous. Thus, no general conclusion regarding the role of coherence can be drawn in that regime.
\end{remark}

\end{proof}
Next we consider the case when the environment initial state is coherent, and that of the system is incoherent.
\subsection{Initial States with Environment Coherence Alone}
\label{subsec:CBenvironmentcoh}
For $a=0$, and for a fixed set of initial parameters $b$, $\tau_1$, and $t$, the presence of environment coherence always enhances the change in ergotropy. In particular, denoting by $\Delta R$ the change in ergotropy for an incoherent environment, $b=0$, and by $\Delta\tilde{R}$ the corresponding quantity in the presence of environment coherence, we have
$\Delta\tilde{R}-\Delta R=h_1\left(\sqrt{\Delta'^2+4|\lambda_2|^2\tau_{01}^2}-|\Delta'|\right)$.
Hence, for $\tau_{01}\neq0$, one has $\Delta R\leq\Delta\tilde{R}$, with equality only when $|\lambda_2|=0$. Since the unitary is arbitrary, one can in general choose a unitary for which $|\lambda_2|\neq0$, demonstrating that environment coherence can provide an additional contribution to the ergotropic gain.

We next determine the maximum gain attainable when the environment is allowed to possess coherence while the initial system state remains incoherent. The optimization is performed over all unitaries of the form given in Eq.~\eqref{eq:U_CB}.

\begin{proposition}
For the choice of initial states such that $a=0$, $\tau_{01}\neq0$ and $p\geq\tau_{11}$, the ergotropic gain in the CB regime is bounded by
\begin{equation}
    \Delta\tilde{R}\leq f\left(\tilde{\Delta},\sin(2\theta_c)\tau_{01}\right),
\end{equation}
where
$\tilde{\Delta}=1-2\Big(\tau_{11}+\cos^2(2\theta_c)(p-\tau_{11})\Big)$.
\label{prop:CBenvironmentcoh}
\end{proposition}

\begin{proof}
For $a=0$, the ergotropic gain in the presence of environment coherence can be written as
\begin{equation*}
\Delta\tilde{R}=h_1\left[f(\Delta',|\lambda_2|\tau_{01})-f(\Delta,0)\right].
\end{equation*}
Using the explicit form of $f$, this becomes
$\Delta\tilde{R}=h_1\left(\Delta'+\sqrt{\Delta'^2+4|\lambda_2|^2\tau_{01}^2}-\Delta-|\Delta|\right)$.

Let $V_{12}=|V_{12}|e^{\iota c_1}$ and $V_{21}=|V_{12}|e^{\iota c_2}$. From the expression for $\lambda_2$, we obtain
\begin{equation*}|\lambda_2|=|V_{12}|\sqrt{(1-p)^2+p^2+2p(1-p)\cos(m)}.\end{equation*}
Where $m=\phi_0+\phi_1-c_1-c_2.$ Consequently,
$|\lambda_2|\leq |V_{12}|$,
with equality when $\phi_0+\phi_1=c_1+c_2$ modulo $2\pi$.

Furthermore, using $|V_{11}|^2=1-|V_{12}|^2$, we have
$\Delta'=1-2\left[\tau_{11}+(1-|V_{12}|^2)(p-\tau_{11})\right]$.
For $p\geq\tau_{11}$, both $\Delta'$ and $|\lambda_2|$ are nondecreasing functions of $|V_{12}|$. Since the quantity
$f\left(\Delta',|\lambda_2|\tau_{01}\right)
=\Delta'+\sqrt{\Delta'^2+4|\lambda_2|^2\tau_{01}^2}$
is monotonically increasing in both arguments over the relevant parameter regime; the ergotropic gain is maximized by maximizing $|V_{12}|$ and $|\lambda_2|$.

For the unitary in Eq.~\eqref{eq:U_CB}, the magnitude of the off-diagonal matrix element satisfies
$|V_{12}|\leq\sin(2\theta_c)$,
with equality for $\phi_+-\phi_-=n\pi$, where $n$ is an odd integer. Therefore,
$|\lambda_2|\leq |V_{12}|\leq\sin(2\theta_c)$.

The maximum is consequently obtained by setting
$|V_{12}|=|\lambda_2|=\sin(2\theta_c)$.
Substituting these values into $\Delta'$ gives
$\tilde{\Delta}=1-2\left[\tau_{11}+\cos^2(2\theta_c)(p-\tau_{11})\right]$,
and hence
$\Delta\tilde{R}\leq f\left(\tilde{\Delta},\sin(2\theta_c)\tau_{01}\right)$.
\end{proof}
\begin{figure}[t]
  \includegraphics[width=0.48\textwidth]{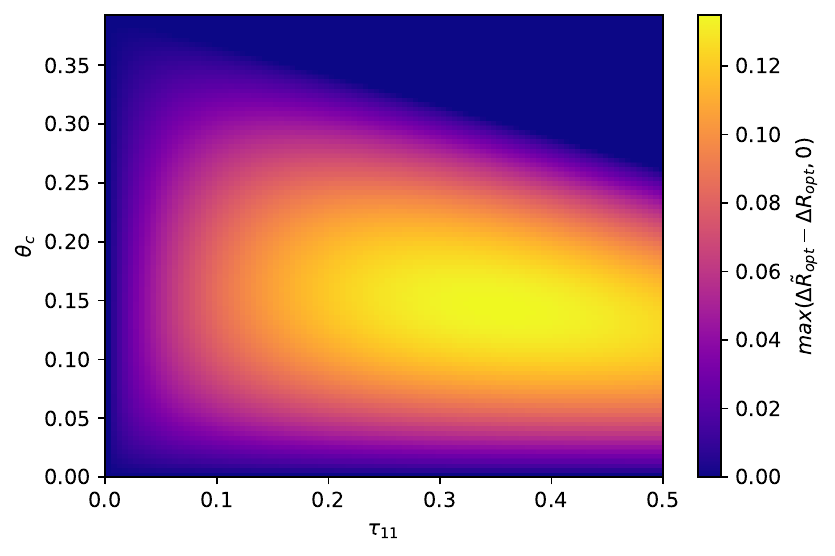}
  \caption{Difference between the ergotropic gain in the presence of environment coherence and the optimal gain for an incoherent environment, $\max(\Delta\tilde{R}{\mathrm{opt}}-\Delta R{\mathrm{opt}}, 0)$, for the initial conditions $p=1/2$ and $a=0$, as a function of $\tau_{11}\in[0,0.5]$ and $\theta_c\in[0,0.35]$. Here, the environment coherence is taken to be maximal, $|\tau_{01}|=\sqrt{\tau_{11}(1-\tau_{11})}$. The enhancement is largest for intermediate values of $\tau_{11}$ and $\theta_c$, demonstrating that environment coherence can substantially enhance the optimal ergotropic gain compared with the incoherent-environment case.}
  \label{fig:probemap}
\end{figure}
\begin{remark}
The proposition applies to both the CB-degenerate and CB-nondegenerate regimes. The difference lies only in the range of $|V_{12}|$ accessible in the two regimes. In the CB-nondegenerate regime, the maximum value of $|V_{12}|$ is determined by $\sin(2\theta_c)$, which is itself constrained by the Hamiltonian parameters $J$, $h_1$, and $h_2$. In the CB-degenerate regime, there is no such additional restriction, and one can in particular choose $|V_{12}|=1$.
\end{remark}

The above proposition allows us to compare the maximum ergotropic gain in the presence of environment coherence with the optimal gain attainable for an incoherent environment. From Proposition~\ref{prop1}, the latter is maximized for $p=\frac{1}{2}$ and $\tau_{11}=0$, yielding
$\Delta R_{\mathrm{opt}}=2h_1\sin^2(2\theta_c)$.

For a fixed population $\tau_{11}$, positivity of the environmental density matrix imposes the constraint
$|\tau_{01}|^2\leq\tau_{11}(1-\tau_{11})$.
Thus, the maximum allowed coherence is
$|\tau_{01}|=\sqrt{\tau_{11}(1-\tau_{11})}$.
Using this value in the bound of Proposition~\ref{prop:CBenvironmentcoh}, and setting $p=1/2$, gives
$\Delta\tilde{R}_{\mathrm{opt}}
=f\left(\tilde{\Delta},\sin(2\theta_c)\sqrt{\tau_{11}(1-\tau_{11})}\right)$.

Figure~\ref{fig:probemap} shows the difference $\max(\Delta\tilde{R}_{\mathrm{opt}}-\Delta R_{\mathrm{opt}}, 0)$ as a function of $\tau_{11}$ and $\theta_c$. The enhancement is concentrated in an extended interior region of the parameter space, with the largest values occurring for intermediate $\tau_{11}$ and moderate values of $\theta_c$. In particular, the enhancement is most pronounced around $\theta_c\approx0.15$ and $\tau_{11}\approx0.3$--$0.4$, where the color map reaches its maximum intensity. The enhancement gradually decreases as either parameter approaches the boundary of the allowed region. For $\theta_c\rightarrow0$, the factor $\sin(2\theta_c)$ vanishes, and consequently the contribution of environment coherence disappears. Similarly, the enhancement becomes small when $\tau_{11}$ approaches the endpoints, since the maximum physically allowed coherence $\sqrt{\tau_{11}(1-\tau_{11})}$ vanishes as $\tau_{11}\rightarrow0$ or $\tau_{11}\rightarrow1$. The figure therefore demonstrates that environment coherence can provide a substantial enhancement of the ergotropic gain over the optimal incoherent-environment value, although the magnitude of the enhancement depends sensitively on both the environment coherence and the coherent coupling parameter $\theta_c$.

This behavior is qualitatively different from the case of an initially coherent system coupled to an incoherent environment. In that case, Proposition~\ref{prop2} shows that system coherence does not enhance the optimal ergotropic gain; rather, the presence of system coherence reduces the maximum gain attainable from the corresponding incoherent initial state.

Once we have analyzed the CB regime. In the following subsection we analyze the ergotropic gain in DD regime.

\section{Ergotropic Injection at the Double Degeneracy}
\label{sec:DD}

At $J=J_c=2\sqrt{h_1h_2}$, the Hamiltonian spectrum collapses into two doubly degenerate levels, $\{E,E,-E,-E\}$, with $E=h_1+h_2$. The positive-energy eigenspace is spanned by ${\ket{00},\ket{+}}$, while the negative-energy eigenspace is spanned by ${\ket{-},\ket{11}}$, where $\ket{\pm}$ are the dressed states of the central block defined in Eq.~\eqref{eq:dressed}. The corresponding angle satisfies
\begin{equation}
\cos^2\theta_c=\frac{h_1}{h_1+h_2},
\qquad
\sin^2\theta_c=\frac{h_2}{h_1+h_2}.
\label{eq:DDangle}
\end{equation}

As discussed in Sec.~\ref{subsec:classification}, an admissible energy-conserving unitary in this regime has the block form $U=W_+\oplus W_-$, where $W_+\in U(2)$ acts on the positive-energy subspace ${\ket{00},\ket{+}}$ and $W_-\in U(2)$ acts on the negative-energy subspace ${\ket{-},\ket{11}}$. Denoting the matrix elements of $W_+$ and $W_-$ by $A_{ij}=(W_+)_{ij}$ and $B_{kl}=(W_-)_{kl}$, respectively, and transforming back to the computational basis, we obtain
\begin{equation}
U=
\begin{pmatrix}
A_{11} & cA_{12} & sA_{12} & 0\\
cA_{21} & c^2A_{22}+s^2B_{11} & cs(A_{22}-B_{11}) & -sB_{12}\\
sA_{21} & cs(A_{22}-B_{11}) & s^2A_{22}+c^2B_{11} & cB_{12}\\
0 & -sB_{21} & cB_{21} & B_{22}
\end{pmatrix},
\label{eq:doubleU}
\end{equation}
where $c=\cos\theta_c$ and $s=\sin\theta_c$. Unlike the central-block dynamics, Eq.~\eqref{eq:doubleU} can couple different computational sectors while conserving energy. This is a consequence of the fact that the dressed state $\ket{+}$ is a superposition of $\ket{01}$ and $\ket{10}$. For example, the matrix element $\bra{01}U\ket{00}=cA_{21}$ couples the zero- and single-excitation sectors while remaining within the degenerate positive-energy subspace.

The reduced state of the system after the joint unitary evolution is given by
$\sigma_S=\Tr_E[U\rho_{SE}U^\dagger]$,
where $\rho_{SE}=\rho\otimes\rho_E$, with $\rho$ and $\rho_E$ parametrized as in Eqs.~\eqref{eq:qubitparam} and \eqref{env}, respectively. For convenience, we define the population weights of the initial joint state $\rho_{SE}$ in the computational basis as
\begin{align*}
w_{00} &= (1-p)(1-\tau_{11}), & w_{01} &= (1-p)\tau_{11},\\
w_{10} &= p(1-\tau_{11}), & w_{11} &= p\tau_{11}.
\end{align*}
In the following subsection, we consider initially incoherent input states and derive the exact form of $\sigma_S$, which we subsequently use to determine the change in ergotropy.

\subsection{Ergotropy injection in the DD regime when the initial states are incoherent}

The essential point is that, although the input state is diagonal in the computational basis, it is not diagonal in the eigenbasis ${\ket{00},\ket{+},\ket{-},\ket{11}}$. Using $\ket{01}=c\ket{+}-s\ket{-}$ and $\ket{10}=s\ket{+}+c\ket{-}$, the singly-excited populations populate the dressed levels as
\begin{equation}
\begin{split}
\langle+|\rho_{SE}|+\rangle =\rho_{++}=c^{2}w_{01}+s^{2}w_{10},\\
\langle-|\rho_{SE}|-\rangle =\rho_{--}=s^{2}w_{01}+c^{2}w_{10}.
\label{eq:dressedpops}
\end{split}
\end{equation}
Moreover, they generate a coherence between the two dressed states,
\begin{equation}
\eta \equiv\bra{+}\rho_{SE}\ket{-}
= cs(w_{10}-w_{01})
= \tfrac12\sin2\theta_c(p-\tau_{11}).
\label{eq:eta}
\end{equation}

Using the expression for $U$, the final reduced state $\sigma_{SE}$ of the system can be written in the computational basis as
\begin{equation}
\sigma =
\begin{pmatrix}
\sigma_{00} & \sigma_{01}\\
\sigma_{10} & \sigma_{11}
\end{pmatrix}.
\end{equation}
The population $\sigma_{00}$ is given by
\begin{equation}
\sigma_{00}=P_{00}+c^{2}P_{+}+s^{2}P_{-}-2cs\Re\eta',
\label{eq:DDsigma00}
\end{equation}
where
\begin{align}
P_{00} &= w_{00}|A_{11}|^{2}+\rho_{++}|A_{12}|^{2},
\nonumber\\
P_{+} &= w_{00}|A_{21}|^{2}+\rho_{++}|A_{22}|^{2},
\nonumber\\
P_{-} &= \rho_{--}|B_{11}|^{2}+w_{11}|B_{12}|^{2},
\nonumber\\
\eta' &\equiv \bra{+}\sigma_{SE}\ket{-}
= A_{22}\eta B_{11}^{}.
\label{eq:eigpops}
\end{align}
The off-diagonal element is
\begin{align}
\sigma_{01}=&s\bigl(A_{11}A_{21}^{}w_{00}
+A_{12}A_{22}^{}\rho_{++}\bigr)
\nonumber\\
&-s\bigl(B_{11}B_{21}^{}\rho_{--}
+B_{12}B_{22}^{}w_{11}\bigr)
+c\eta\zeta,
\label{eq:DDsigma01}
\end{align}
with $\zeta=A_{12}B_{11}^{}+A_{22}B_{21}^{*}$. Thus, both the population in Eq.~\eqref{eq:DDsigma00} and the coherence in Eq.~\eqref{eq:DDsigma01} depend on $\eta$. This gives rise to a \emph{population-to-coherence channel} that is active even when both input states are incoherent. Such a channel is absent in the central-block regime and is the key feature underlying the results below.

We now derive the maximum ergotropy gain in the DD regime. The maximization is performed over all incoherent initial states $\rho$ and $\rho_E$, as well as over all energy-conserving joint unitaries of the form given in Eq.~\eqref{eq:U_DD}.

\begin{figure}[t]
  \includegraphics[height=5.5cm, width=0.45\textwidth]{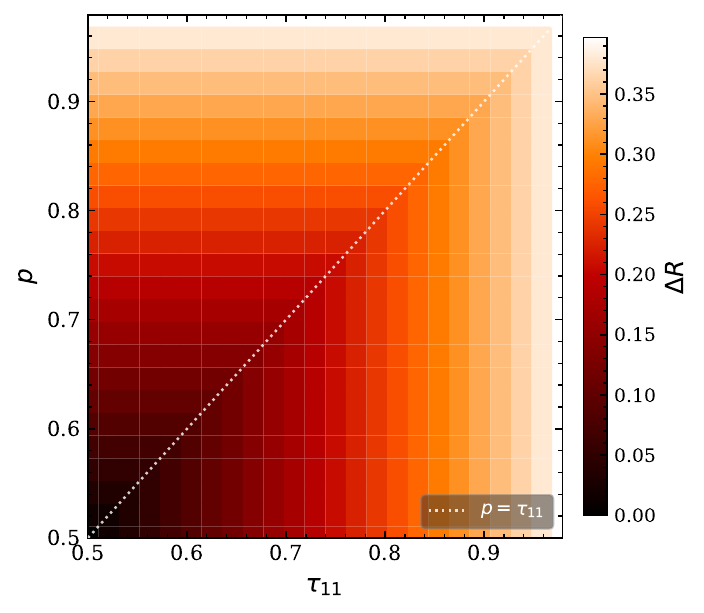}
  \caption{Optimal ergotropic gain $\Delta R$ in the double-degeneracy regime
($J=J_c$) for incoherent inputs, as a function of the environment ground
population $\tau_{11}$ and the system ground population $p$, maximized over
all admissible energy-conserving unitaries $U=W_+\oplus W_-$. The window
shown is the passive quadrant $p,\tau_{11}\ge\tfrac12$, where both qubits
start passive and no ergotropy is present initially. The gain is strictly
positive everywhere except at the single point $p=\tau_{11}=\tfrac12$, where
the joint state is maximally mixed and invariant under every unitary. The
dotted line marks $p=\tau_{11}$; injection persists above it, i.e.\ an
environment more highly populated in its ground state than the system can
still charge the system. }
  \label{fig:prop4map}
\end{figure}
\begin{theorem}
\label{thm:DDfulltransfer}
Suppose that the initial states of the system and environment are incoherent. Then the maximum ergotropy gain $\Delta R$, optimized over all choices of the initial states $\rho$ and $\rho_E$ of the forms given in Eqs.~\eqref{eq:qubitparam} and \eqref{env}, respectively, with $a=0$ and $\tau_{01}=0$, and over all energy-conserving joint unitaries of the form given in Eq.~\eqref{eq:U_DD}, is given by
\begin{equation}
  \max_{p,q,\,U}\Delta R=
  \begin{cases}
    2h_1, & h_2\ge h_1,\\[3pt]
    2h_1\sin^{2}2\theta_c, & h_2<h_1<3h_2,\\[3pt]
    2h_1\,s(1+s), & h_1\ge3h_2 .
  \end{cases}
  \label{eq:DDsup}
\end{equation}
The three branches are attained as follows.
\begin{enumerate}
\item[(i)] For $h_2\ge h_1$ the optimum sits at $(p,\tau_{11})=(1,0)$ and is
reached by full excitation transfer: the pure input $\ket{10}$ is
mapped to a state with the system fully excited, so
$\sigma_S=\ketbra{0}{0}$ and $\Delta R$ saturates the absolute ceiling
$2h_1$.
\item[(ii)] Inside the window $h_2<h_1<3h_2$ the optimum sits at
$(p,\tau_{11})=(\tfrac12,0)$ and is reached by the dressed-diagonal unitary
$W_+=\mathrm{diag}(e^{\iota\alpha},1)$, $W_-=\mathrm{diag}(-1,e^{\iota\beta})$,
which moves population only and creates no coherence.
\item[(iii)] For $h_1\ge3h_2$ the optimum again sits at
$(p,\tau_{11})=(\tfrac12,0)$, now reached by a two-block unitary that acts
identically in the two degenerate sectors: it rotates $\ket{+}$ towards
$\ket{00}$ and $\ket{-}$ towards $\ket{11}$ by the same amplitude, with
$|\bra{00}W_+\ket{+}|^{2}=|\bra{11}W_-\ket{-}|^{2}=1-\frac{1+s}{2c^{2}}$.
Here the gain comes from population transfer and coherence generation
together.
\end{enumerate}
\end{theorem}

\begin{proof}
Ergotropy is convex in the state, and, for fixed $U$, the output state $\sigma_S$ depends affinely on the input state. At fixed $\tau_{11}$, the input is affine in $p$, while at fixed $p$, it is affine in $\tau_{11}$. Hence, $\max_U R(\sigma_S)$ is convex along every horizontal and vertical line in the square $(p,\tau_{11})\in[0,1]^2$. The reference term $R(\rho_S)$ is independent of $\tau_{11}$ and is piecewise affine in $p$, with a kink at $p=1/2$. Consequently, $\max_U\Delta R$ is convex in $\tau_{11}$ on $[0,1]$, and convex in $p$ separately on $[0,1/2]$ and $[1/2,1]$. Since a convex function on a closed segment attains its maximum at an endpoint, it suffices to consider the six candidates $(p,q)\in\{0,\tfrac12,1\}\times\{0,1\}$.

The two candidates with $p=0$ have $R(\rho_S)=2h_1$, which is already the maximum ergotropy that a qubit with energy gap $2h_1$ can possess; hence, no ergotropy gain is possible for these points. The remaining four candidates are solved in closed form in Appendix~\ref{app:DDsup}, with the results summarized in Table~\ref{tab:DDcorners}. Comparing these candidates in the different regimes of $\frac{h_1}{h_2}$, as detailed in Appendix~\ref{app:DDsup}, yields Eq.~\eqref{eq:DDsup}. In particular, the corner $(1,0)$ is optimal for $h_2\ge h_1$, where it achieves the maximal possible gain of $2h_1$, while the corner $(\tfrac12,0)$ is optimal throughout the regime $h_1>h_2$.
\end{proof}

Earlier, in Sec.~\ref{subsec:CBincoherent}, we demonstrated that, in the CB regime, ergotropy injection is not possible even in the presence of interactions in the Hamiltonian, i.e., for $J\neq0$, as long as the environment is initially in a thermal state. In contrast, as a direct consequence of the proof of Theorem~\ref{thm:DDfulltransfer}, we obtain the following corollary, which shows that ergotropy injection with a thermal environment is possible in the DD regime.

\begin{corollary}[Ergotropy from a $T=0$ environment]
\label{thm:DDzeroT}
et $J=J_c$ and suppose that both qubits are initially in their respective ground states, $\rho=\rho_E=\ketbra{1}{1}$. Then, the maximum ergotropic gain, optimized over all admissible energy-conserving unitaries in the DD regime, is given by
\begin{equation}
  \Delta R=
  \begin{cases}
    \;h_1\,\dfrac{1-c}{c}, & c\ge\tfrac12\ \ (h_2\le3h_1),\\[8pt]
    \;2h_1\,(1-2c^{2}), & c<\tfrac12\ \ (h_2>3h_1),
  \end{cases}
  \label{eq:zeroT}
\end{equation}
which is strictly positive for all $h_1,h_2>0$.
\end{corollary}

\begin{proof}
The input is the corner $(p,\tau_{11})=(1,1)$, whose value is the third row of
Table~\ref{tab:DDcorners}, derived in Appendix~\ref{app:DDsup} Sec.~\ref{app:DDsup:11}:
Eq.~\eqref{eq:zeroT}. 
\end{proof}

The above corollary is specific to the initial state with $(p,\tau_{11})=(1,1)$. In the following proposition, we generalize the statement concerning ergotropy injection with a passive environment state. In particular, we show that, for a passive initial system state with $p\geq \frac{1}{2}$, ergotropy injection is possible in the DD regime even when $\tau_{11}>p$. This is in contrast to the CB regime, where ergotropy injection is possible only when $p>\tau_{11}$. Indeed, in the latter case, we have $\Delta R\leq 2h_1(\Delta'-\Delta)$ (see the discussion for $a=0$ in Sec.~\ref{SCCB}). Since $\Delta'-\Delta=2(1-t)(p-\tau_{11})$, it follows that $\Delta R\leq0$ whenever $\tau_{11}\geq p$. The proposition is stated as follows.
\begin{figure*}[t]
  \includegraphics[width=\textwidth]{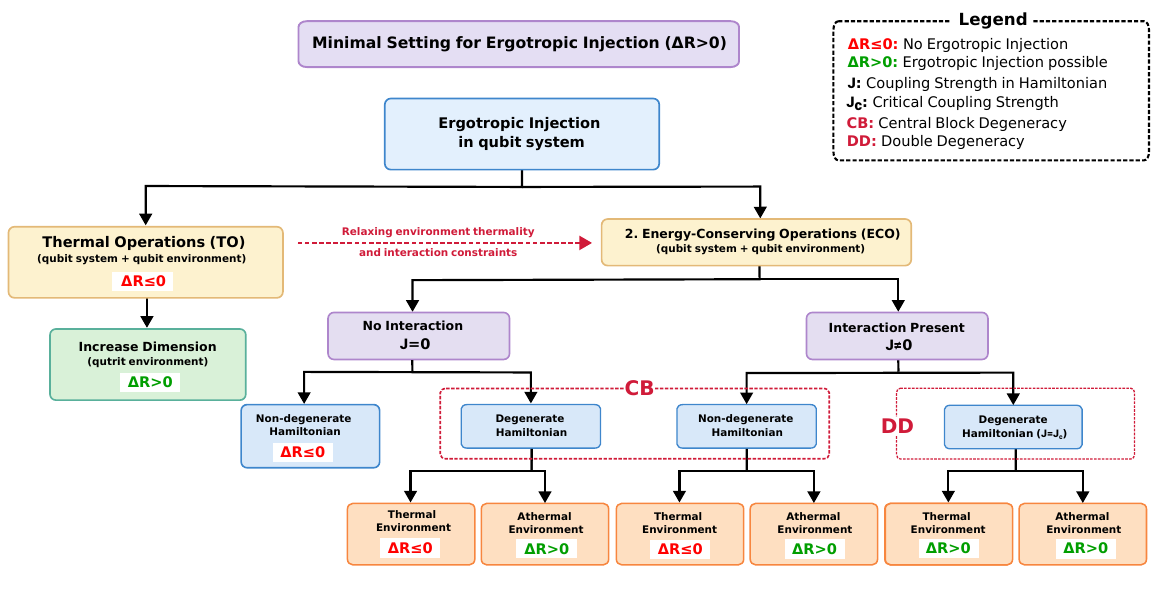}
  \caption{Summary of the ergotropic injection regimes considered in this work. Under thermal operations with a qubit system and a qubit environment, no injection is possible. Two routes circumvent this. Enlarging the environment to a qutrit restores injection within the thermal-operation framework. Alternatively, one may retain qubits and pass to the broader class of energy-conserving operations, where the unitary commutes with the full Hamiltonian including any interaction term and the environment need not be thermal. Injection is then governed by the degeneracy structure of the total Hamiltonian: for a nondegenerate non-interacting Hamiltonian it remains impossible; in the central-block regime, whether degenerate ($J=0$, $h_1=h_2$) or nondegenerate ($J\neq0$, $J\neq J_c$), it requires an athermal environment; and at the interaction-induced double degeneracy $J=J_c$ it occurs for a thermal environment as well.}
  \label{fig:flowchart}
\end{figure*}
\begin{proposition}[Injection for every passive system state]
\label{prop:DDpq}
For incoherent input states $\rho$ and $\rho_E$ with $p\ge\tfrac12$, such that $\rho$ is passive, ergotropy is injected for every $(p,\tau_{11})$ in the DD regime ($J=J_C$), except at the single trivial point $p=\tau_{11}=1/2$. In particular, ergotropy injection occurs even when $\tau_{11}>p$: an environment that is more highly populated in its ground state than the system can still charge it.
\end{proposition}

\begin{proof}
For $p\ge\tfrac12$ we have $\Delta=1-2p\le0$, so $R(\rho_S)=0$ and
$\Delta R=R(\sigma_S)
=h_1\bigl[\Delta'+\sqrt{\Delta'^{2}+4|\sigma_{01}|^{2}}\bigr]$. This expression is strictly positive for any $|\sigma_{01}|>0$, regardless of the sign of $\Delta'$, because the square root strictly exceeds $|\Delta'|$. 
By Eq.~\eqref{eq:DDsigma01}, we can choose a unitary such that $A_{11}=0$, which gives, $\sigma_{01} = sB_{11}B_{21}^*(2p\tau_{11}-s^2\tau_{11}-c^2p)+c^2s(p-\tau_{11})A_{12}B_{11}^*$. Since the $A$ and $B$ blocks are independent, their admissible matrix elements, including both their magnitudes and phases, can be varied subject to the unitarity constraints. Thus, if a particular choice of $A$ and $B$ yields $\sigma_{01}=0$, this does not imply that $\sigma_{01}$ vanishes for all admissible choices. Unless the coefficients of both contributions vanish identically, the matrix elements can be chosen such that at least one contribution is nonzero. Moreover, when both contributions are nonzero, their magnitudes and relative phase can be varied to prevent cancellation. Hence, one can choose admissible $A$ and $B$ such that $\sigma_{01}\neq0$.

The two contributions can vanish identically only if both coefficients vanish simultaneously, i.e., $2p\tau_{11}-s^2\tau_{11}-c^2p=0$ and $p-\tau_{11}=0$. Setting $p=\tau_{11}$ in the first condition and using $s^2+c^2=1$ gives $p(2p-1)=0$. Thus, the only solutions are $(p,\tau_{11})=(0,0)$ and $(p,\tau_{11})=(\tfrac12,\tfrac12)$. The former lies outside the theorem's relevant domain. At $(p,\tau_{11})=(\tfrac12,\tfrac12)$, the initial state is maximally mixed, $\rho_{SE}=I_4/4$, which remains invariant under every unitary. Consequently, $\sigma_{01}$ vanishes identically for every admissible unitary, and no coherent contribution can be generated.

In Fig.~\ref{fig:prop4map}, we present the ergotropic gain in the double-degenerate regime for initially incoherent system and environment states with $p,\tau_{11}>1/2$. To determine the ergotropic gain, we optimize over all admissible energy-conserving unitaries of the form $U=W_+\oplus W_-$. For the optimization, we employed the gradient-free Nelder–Mead algorithm from the \texttt{scipy.optimize} Python package ~\cite{2020SciPy-NMeth}.  The plot clearly complements the above proposition: the gain is strictly positive everywhere except at the single point $p=\tau_{11}=\tfrac12$, where the joint state is maximally mixed and hence invariant under any unitary. The dotted line marks $p=\tau_{11}$; notably, ergotropic injection persists above this line, showing that an environment with a higher ground-state population than the system can still activate the system.

\end{proof}


\section{Summary Of Results}
\label{summary}
In Fig.~\ref{fig:flowchart}, we have given a flowchart that systemetically depicts our result in this work, Summary of the ergotropic injection regimes considered in this work. Under thermal operations with a qubit system and a qubit environment, no injection is possible. Two routes circumvent this. Enlarging the environment to a qutrit restores injection within the thermal-operation framework. Alternatively, one may retain qubits and pass to the broader class of energy-conserving operations, where the unitary commutes with the full Hamiltonian including any interaction term and the environment need not be thermal. Injection is then governed by the degeneracy structure of the total Hamiltonian: for a nondegenerate non-interacting Hamiltonian it remains impossible; in the central-block regime, whether degenerate ($J=0$, $h_1=h_2$) or nondegenerate ($J\neq0$, $J\neq J_c$), it requires an athermal environment; and at the interaction-induced double degeneracy $J=J_c$ it occurs for a thermal environment as well.
\section{Conclusion}
\label{sec:conclusion}We have investigated the minimal physical setting for injecting ergotropy into a quantum system through its interaction with an environment. We first established that, within the framework of thermal operations, a qubit environment was insufficient to inject ergotropy in a qubit system, irrespective of the environmental temperature and energy gap. By contrast, we showed that increasing the environmental dimension to a qutrit was sufficient to activate a positive ergotropic gain. Thus, the qutrit constituted the minimal environmental dimension capable of ergotropy injection into a qubit under thermal operations.

We then went beyond thermal operations while retaining strict energy conservation by allowing an interaction term in the total Hamiltonian and considering arbitrary initial environment states. Using a two-qubit isotropic XY model, we classified the admissible energy-conserving unitaries according to the degeneracy structure of the total Hamiltonian. In the central-block regime, we found that, for initially incoherent system and environment states, the interaction alone could not generate ergotropy, and environmental athermality was necessary for a positive ergotropic gain. We derived the optimal gain and showed that environmental coherence could further enhance it, whereas coherence initially present only in the system did not necessarily provide an advantage and could instead reduce the gain.

We further analyzed a distinct double-degeneracy regime and found that positive ergotropic gain could arise even when the environment was initially passive. In particular, we showed that ergotropy could be injected even when the environment was more passive than the initial state of the system. We also derived the optimal ergotropic gain in this regime. These results demonstrated that the ability of an environment to inject ergotropy was determined not solely by its initial energetic properties, but also by the degeneracy structure of the interacting system-environment Hamiltonian and the resulting set of energy-conserving unitaries.

Overall, our results established the minimal environmental dimension required for ergotropy injection under thermal operations and showed how this restriction could be circumvented in the minimal qubit-qubit setting by relaxing environmental thermality and allowing interactions in the total Hamiltonian. They further clarified the distinct roles played by environmental athermality, coherence, and Hamiltonian interaction in enabling and enhancing ergotropy injection. These findings provided a finite-dimensional setting for identifying the physical resources responsible for charging an open quantum system and highlighted the importance of the Hamiltonian structure, beyond the initial state of the environment alone, in determining the extractable work generated through system-environment interactions.
\section{Acknowledgment}
US acknowledges financial support from the Anusandhan National Research Foundation (ANRF), Government of India, under the Grant No. ANRF/ARG/2025/004617/PS.
\begin{appendix}
\onecolumngrid
    \section{Complete proof of Theorem 4}
\label{app:DDsup}

Throughout this appendix $J=J_c$, $c^{2}=\frac{h_1}{h_1+h_2}$, and
$s^{2}=\frac{h_2}{h_1+h_2}$. The proof runs in three steps: a convexity
argument that reduces the optimization over inputs to six corner
states (Sec.~\ref{app:DDsup:convex}); an exact evaluation of the gain
at each corner (Secs.~\ref{app:DDsup:half}--\ref{app:DDsup:11}); and
a comparison of the corners in each regime of $\frac{h_1}{h_2}$
(Sec.~\ref{app:DDsup:compare}).\\

Given an incoherent input system state $\rho_S=(1-p)\ketbra{0}{0}+p\ketbra{1}{1}$ and environment state $\rho_e=\tau_{00}\ketbra{0}{0}+\tau_{11}\ketbra{1}{1}$, the input state ergotropy is given by
$R(\rho_S)=2h_1\max(1-2p,0)$. Let the output state be $\sigma_S$ with excited state population $\sigma_{00}$ and off-diagonal term $\sigma_{01}$. Writing $\Delta'=2\sigma_{00}-1$ for the
output state, the quantity to be maximized is the ergotropic gain
\begin{equation}
  \Delta R=h_1\Bigl[\Delta'+\sqrt{\Delta'^{2}+4|\sigma_{01}|^{2}}
  -2\max(1-2p,0)\Bigr]
  \label{eq:app_DR}
\end{equation}

\subsection{Reduction to six corner inputs}
\label{app:DDsup:convex}

We first recall what \emph{affine} means. A quantity $X$ depends
affinely on a parameter $\lambda$ if $X(\lambda)=X_0+\lambda X_1$ with
$X_0,X_1$ independent of $\lambda$, i.e.\ linear plus a constant and
no higher powers.\\

\emph{(1) The input is affine in $p$ and in $\tau_{11}$ separately:} In the
computational basis $\{\ket{00},\ket{01},\ket{10},\ket{11}\}$ the
incoherent product input is
\begin{equation*}
  \rho_{SE}(p,\tau_{11})=\mathrm{diag}\bigl((1-p)(1-\tau_{11}),\,(1-p)\tau_{11},\,
  p(1-\tau_{11}),\,p\tau_{11}\bigr).
  \label{eq:app_input}
\end{equation*}
At fixed $\tau_{11}$, every entry is of the form $X_0+pX_1$, and at fixed
$p$ every entry is of the form $X_0'+\tau_{11}X_1'$. So $\rho_{SE}$ is
affine along every horizontal and every vertical line of the square
$[0,1]^{2}$. \\

\emph{(2) The output is affine along the same lines:} For fixed
admissible $U$,
$\sigma_S=\tr_E[U\rho_{SE}U^{\dagger}]$ is obtained from $\rho_{SE}$
by unitary conjugation followed by a partial trace, both linear maps.
Composing a linear map with an affine dependence leaves it affine, so the maps
$p\mapsto\sigma_S$ at fixed $q$, and $\tau_{11}\mapsto\sigma_S$ at fixed $p$,
are affine.\\

\emph{(3) Ergotropy is convex in the state:} From
$R(\rho)=\Tr(H_S\rho)-\min_{V}\Tr(H_SV\rho V^{\dagger})$, the first
term is linear in $\rho$ and the second is a minimum of linear
functionals of $\rho$. A pointwise minimum of linear functions is
concave, so its negative is convex, and $R=(\text{linear})+
(\text{convex})$ is convex; this follows directly
from the variational definition of Ref.~\cite{Allahverdyan2004}.\\

\emph{(4) The optimized output ergotropy is convex along those lines:}
Composing a convex function with an affine map preserves convexity, so
for each fixed $U$ the map $p\mapsto R(\sigma_S)$ at fixed $\tau_{11}$ (and
$\tau_{11}\mapsto R(\sigma_S)$ at fixed $p$) is convex. The admissible set
$\{U=W_+\oplus W_-\}\cong\mathrm{U}(2)\times\mathrm{U}(2)$ is compact
and $R(\sigma_S)$ is continuous in $U$, so the maximum is attained;
and a pointwise maximum of convex functions is convex. Hence
$F(p,\tau_{11})\equiv\max_U R(\sigma_S)$ is convex in $p$ at fixed $\tau_{11}$ and
convex in $\tau_{11}$ at fixed $p$. (For these standard
properties of convex functions, see e.g.~\cite{Boyd2004}).\\

\emph{(5) The reference term $p={1}/{2}$ as a corner candidate:} We
must subtract $R(\rho_S)=2h_1\max(1-2p,0)$, which is constant in $\tau_{11}$ but
\emph{convex} in $p$; subtracting a convex function does not preserve
convexity, so we split its two affine pieces:
\begin{equation}
  R(\rho_S)=
  \begin{cases}
    2h_1(1-2p), & 0\le p\le\frac12,\\
    0, & \frac12\le p\le1 .
  \end{cases}
\end{equation}
On each piece $R(\rho_S)$ is affine, so
$\Delta R^{*}(p,\tau_{11})\equiv F(p,\tau_{11})-R(\rho_S)$ is convex in $p$ separately
on $[0,\tfrac12]$ and on $[\frac12,1]$. The kink of $R(\rho_S)$ at
$p=\frac12$ is the sole reason $p=\frac12$ enters the candidate
list: it is the junction of the two convexity intervals, hence an
admissible endpoint for both. In $\tau_{11}$, at fixed $p$, no such splitting
is needed and $\Delta R^{*}$ is convex on all of $[0,1]$.\\

\emph{(6) Reduction:} A convex function on a segment attains its
maximum at an endpoint. Maximizing first in $\tau_{11}$ at fixed $p$ gives
$\tau_{11}\in\{0,1\}$; maximizing then in $p$ at fixed $\tau_{11}$, on each convexity
interval, gives $p\in\{0,\frac12,1\}$. Therefore
\begin{equation}
  \max_{p,\tau_{11}}\Delta R^{*}
  =\max\Bigl\{\Delta R^{*}(p,\tau_{11}):\,
  p\in\{0,\tfrac12,1\},\ \tau_{11}\in\{0,1\}\Bigr\}.
\end{equation}
The two corners with $p=0$ have $\rho_S=\ketbra{0}{0}$ and
$R(\rho_S)=2h_1$, which is the maximum ergotropy of a qubit with gap $2h_1$.
Since $R(\sigma_S)\le2h_1$ always, they give $\Delta R\le0$ and are
discarded. Four candidates remain:
\begin{equation}
  (p,\tau_{11})\in\bigl\{(1,0),\ (\tfrac12,0),\ (\tfrac12,1),\ (1,1)\bigr\}.
\end{equation}

We use the dressed basis $\{\ket{00},\ket{+},\ket{-},\ket{11}\}$ with
$\ket{+}=c\ket{01}+s\ket{10}$ and $\ket{-}=-s\ket{01}+c\ket{10}$,
equivalently
\begin{equation}
  \ket{01}=c\ket{+}-s\ket{-},\qquad \ket{10}=s\ket{+}+c\ket{-}.
  \label{eq:app_basis}
\end{equation}
Admissible unitaries are $U=W_+\oplus W_-$ with
\begin{align*}
  W_+&=\begin{pmatrix}x_0&x\\ y_0&y\end{pmatrix}
  \ \text{on }(\ket{00},\ket{+}),
  \\ \nonumber
  W_-&=\begin{pmatrix}u&\tilde u\\ v&\tilde v\end{pmatrix}
  \ \text{on }(\ket{-},\ket{11}),
\end{align*}
so that
\begin{align}
  U\ket{00}&=x_0\ket{00}+y_0\ket{+}, & U\ket{+}&=x\ket{00}+y\ket{+},\nonumber\\
  U\ket{-}&=u\ket{-}+v\ket{11}, &
  U\ket{11}&=\tilde u\ket{-}+\tilde v\ket{11}.
  \label{eq:app_images}
\end{align}
Unitarity of a $2\times2$ block fixes the moduli pairwise and makes
the rows orthogonal:
\begin{gather}
  |x_0|=|y|,\quad |y_0|=|x|,\quad x_0\bar y_0=-x\bar y,
  \label{eq:app_unitarityP}\\
  |\tilde v|=|u|,\quad |\tilde u|=|v|,\quad
  \tilde u\bar{\tilde v}=-u\bar v .
  \label{eq:app_unitarityM}
\end{gather}
These two identities are what eliminate the unused columns below. 


\subsection{The corners \texorpdfstring{$(\tfrac12,0)$ and
$(\tfrac12,1)$}{(1/2,0) and (1/2,1)}}
\label{app:DDsup:half}

Both corners have $p=\tfrac12$, hence $R(\rho_S)=0$, and differ only
in which environment state is occupied. Their algebra is identical up
to one sign, so we treat them together and set
\begin{equation}
  \varepsilon=+1\ \ (\tau_{11}=0),\qquad \varepsilon=-1\ \ (\tau_{11}=1).
\end{equation}

\noindent
\emph{Step 1:The two input branches:} For $\tau_{11}=0$ the input is
$\rho_{SE}=\tfrac12(\ketbra{00}{00}+\ketbra{10}{10})$ and, using
Eq.~\eqref{eq:app_basis},
\begin{equation}
  U\ket{10}=sx\ket{00}+sy\ket{+}+cu\ket{-}+cv\ket{11}.
  \label{eq:app_U10}
\end{equation}
For $\tau_{11}=1$ the input is
$\rho_{SE}=\tfrac12(\ketbra{01}{01}+\ketbra{11}{11})$ and
\begin{equation}
  U\ket{01}=cx\ket{00}+cy\ket{+}-su\ket{-}-sv\ket{11}.
  \label{eq:app_U01}
\end{equation}
The second branch is $U\ket{00}$ for $\tau_{11}=0$ and $U\ket{11}$ for $\tau_{11}=1$,
both given in Eq.~\eqref{eq:app_images}.\\

\noindent
\emph{Step 2: Excited population:} Take $\tau_{11}=0$ first. From
Eq.~\eqref{eq:app_U10}, the excited population
carried by $U\ket{10}$ is
\begin{equation}
s^{2}|x|^{2}+c^{2}s^{2}|y|^{2}+ s^{2}c^{2}|u|^{2} -2c^{2}s^{2}\Re(y\bar u)
  =s^{2}|x|^{2}+c^{2}s^{2}|Z|^{2},
\end{equation}
where the cross term comes from $\tr_E\ketbra{+}{-}$ and $Z\equiv y-u$. The branch $U\ket{00}$ contributes $|x_0|^{2}+c^{2}|y_0|^{2}$. Adding
the two with weight $\tfrac12$ and using $|x_0|=|y|$, $|y_0|=|x|$ from
Eq.~\eqref{eq:app_unitarityP},
\begin{align}
  \sigma_{00}
  &=\tfrac12\bigl[|y|^{2}+c^{2}|x|^{2}+s^{2}|x|^{2}\bigr]
  +\tfrac12c^{2}s^{2}|Z|^{2}\nonumber\\
  &=\tfrac12+\tfrac12c^{2}s^{2}|Z|^{2},
  \label{eq:app_half_hot_sig00}
\end{align}
because $|x|^{2}+|y|^{2}=1$ and $c^{2}+s^{2}=1$. For $\tau_{11}=1$ the same
bookkeeping applied to Eq.~\eqref{eq:app_U01}, and 
$|\tilde u|=|v|$, gives
\begin{equation}
  \sigma_{00}=\tfrac12-\tfrac12c^{2}s^{2}|Z|^{2}.
  \label{eq:app_half_cold_sig00}
\end{equation}
In both cases
\begin{equation}
  \Delta'=2\sigma_{00}-1=\varepsilon\,c^{2}s^{2}|Z|^{2}.
  \label{eq:app_half_Delta}
\end{equation}
The interpretation is direct: with an inverted environment ($\tau_{11}=0$) the
protocol imports excitation, with a ground-state environment ($\tau_{11}=1$)
it must pay population out; the magnitude is the same and is governed
by the single variable $|Z|$.\\

\noindent
\emph{Step 3: Coherence:} For $\tau_{11}=0$, Eq.~\eqref{eq:app_U10} give
$s^{3}x\bar y+c^{2}sx\bar u+c^{2}s(y-u)\bar v$ from that branch, and
$sx_0\bar y_0$ from $U\ket{00}$. Eliminating the latter with
$x_0\bar y_0=-x\bar y$ and using $s-s^{3}=sc^{2}$,
\begin{equation}
  \sigma_{01}=\frac{sc^{2}}{2}\bigl[\bar vZ-x\bar Z\bigr].
  \label{eq:app_half_hot_sig01}
\end{equation}
For $q=1$ the same steps, now eliminating $\tilde u\bar{\tilde v}$
with Eq.~\eqref{eq:app_unitarityM}, give
\begin{equation}
  \sigma_{01}=\frac{sc^{2}}{2}\bigl[x\bar Z-\bar vZ\bigr],
  \label{eq:app_half_cold_sig01}
\end{equation}
which differs only by an overall sign. Hence, in both cases
\begin{equation}
  |\sigma_{01}|=\frac{sc^{2}}{2}\bigl|\bar vZ-x\bar Z\bigr|
  \le\frac{sc^{2}}{2}\bigl(|x|+|v|\bigr)|Z| .
  \label{eq:app_half_sig01_bound}
\end{equation}

\noindent
\emph{Step 4: Fixing the free phases:} The second column of $W_+$ and
the first column of $W_-$ may be any unit vectors, so the four phases
$\arg x,\arg y,\arg u,\arg v$ are free. Only the relative phase of $y$
and $u$ enters $|Z|$; the phases of $x$ and $v$ enter only
Eq.~\eqref{eq:app_half_sig01_bound}. We may therefore choose them
independently: first set $\arg u=\arg y+\pi$, which gives
$|Z|=|y|+|u|$; then choose $\arg x$ and $\arg v$ so that the two terms
of Eq.~\eqref{eq:app_half_sig01_bound} add rather than cancel,
saturating the bound. With
\begin{equation}
  A\equiv|Z|,\qquad B\equiv|x|+|v|,
\end{equation}
Eqs.~\eqref{eq:app_half_Delta} and \eqref{eq:app_half_sig01_bound}
inserted into Eq.~\eqref{eq:app_DR} give
\begin{equation}
  \Delta R=h_1c^{2}s\,A\Bigl[\sqrt{s^{2}A^{2}+B^{2}}+\varepsilon sA\Bigr]
  \equiv h_1c^{2}s\;g_\varepsilon(A).
  \label{eq:app_half_gain}
\end{equation}

\noindent
\emph{Step 5: Gain Monotonicity with $A$:} Write
$\mathcal D=\sqrt{s^{2}A^{2}+B^{2}}\ge sA$. Then
$g_\varepsilon(A)=A(\mathcal D+\varepsilon sA)$ and
\begin{equation}
  g_\varepsilon'(A)=\mathcal D+2\varepsilon sA+\frac{s^{2}A^{2}}{\mathcal D}
  =\frac{(\mathcal D+\varepsilon sA)^{2}}{\mathcal D}\;\ge\;0 ,
  \label{eq:app_half_monotone}
\end{equation}
for both signs. The gain is therefore non-decreasing in $A$ at fixed
$B$, which is why the phase choice $|Z|=|y|+|u|$ above is optimal.\\

\noindent
\emph{Step 6: Reduction to one angle:} Put
\begin{equation}
  |y|=\cos\omega,\ |x|=\sin\omega,\;
  |u|=\cos\mu,\ |v|=\sin\mu,
\end{equation}
with $\omega,\mu\in[0,\tfrac{\pi}{2}]$, so that
\begin{equation}
  \vec P\equiv(A,B)
  =(\cos\omega+\cos\mu,\ \sin\omega+\sin\mu)
\end{equation}
is the sum of two unit vectors in the first quadrant, with
$|\vec P|^{2}=2+2\cos(\omega-\mu)\le4$. Equation
\eqref{eq:app_half_gain} is homogeneous equation of degree two in $\vec P$:
replacing $\vec P\to\lambda\vec P$ multiplies $\Delta R$ by
$\lambda^{2}$. Since $\Delta R\ge0$, along every fixed direction the
gain grows with $|\vec P|$, so the optimum lies on the outer boundary
$|\vec P|=2$, which is reached only at $\omega=\mu$. (Every
direction in the quadrant is attainable there, with
$\vec P=2(\cos\omega,\sin\omega)$.) Setting
\begin{equation}
  \nu=\cos^{2}\omega\in[0,1],
\end{equation}
we have $A=2\cos\omega$, $B=2\sin\omega$, and
$A^{2}(s^{2}\cos^{2}\omega+\sin^{2}\omega)=4\nu(1-c^{2}\nu)$, so
\begin{equation}
  \Delta R=4h_1c^{2}s\,\Phi_\varepsilon(\nu),
\quad
  \Phi_\varepsilon(\nu)=\sqrt{\nu-c^{2}\nu^{2}}+\varepsilon s\nu .
  \label{eq:app_half_scalar}
\end{equation}

\noindent
\emph{Step 7: Stationarity and root selection:} Differentiating
Eq.~\eqref{eq:app_half_scalar},
\begin{equation}
  \Phi_\varepsilon'(\nu)=\frac{1-2c^{2}\nu}{2\sqrt{\nu-c^{2}\nu^{2}}}+\varepsilon s
  =0
  \ \Longleftrightarrow\
  1-2c^{2}\nu=-2\varepsilon s\sqrt{\nu-c^{2}\nu^{2}} .
  \label{eq:app_half_stat}
\end{equation}
Because the right-hand side has the sign of $-\varepsilon$, a solution
requires
\begin{equation}
  \varepsilon\bigl(2c^{2}\nu-1\bigr)\ge0 .
  \label{eq:app_half_signtest}
\end{equation}
Squaring Eq.~\eqref{eq:app_half_stat} and using $s^{2}=1-c^{2}$ gives
the same quadratic for both signs,
\begin{equation}
  4c^{2}\nu^{2}-4\nu+1=0,
  \qquad
  \nu_\pm=\frac{1\pm s}{2c^{2}},
  \qquad 2c^{2}\nu_\pm=1\pm s .
  \label{eq:app_half_quad}
\end{equation}
The sign test \eqref{eq:app_half_signtest} now selects one root each:
\begin{itemize}
\item $\varepsilon=+1$ (corner $(\tfrac12,0)$): $2c^{2}\nu>1$ is needed, so
$\nu^{*}=\nu_{+}=\frac{1+s}{2c^{2}}$ and $\nu_-$ is spurious. This root lies in
$[0,1]$ iff $1+s\le2c^{2}=\frac{2(1-s)}{1+s}$, i.e.\ iff
\begin{equation}
  s\le\tfrac12
  \quad\Longleftrightarrow\quad
  h_1\ge3h_2 .
\end{equation}
\item $\varepsilon=-1$ (corner $(\tfrac12,1)$): $2c^{2}\nu<1$ is needed, so
$\nu^{*}=\nu_{-}=(1-s)/(2c^{2})$. This root is always admissible, since
$1-s\le2(1-s)(1+s)$ reduces to $1\le2(1+s)$.
\end{itemize}

\noindent
\emph{Step 8: Optimum values:} At either root,
\begin{equation}
  \nu^{*}-c^{2}\nu^{*2}=\nu^{*}\bigl(1-c^{2}\nu^{*}\bigr)
  =\frac{1+\varepsilon s}{2c^{2}}\cdot\frac{1-\varepsilon s}{2}
  =\frac{1-s^{2}}{4c^{2}}=\frac14 ,
\end{equation}
so $\sqrt{\nu^{*}-c^{2}\nu^{*2}}=\tfrac12$ and
Eq.~\eqref{eq:app_half_scalar} gives
\begin{equation}
  \Delta R=4h_1c^{2}s\Bigl[\frac12
  +\varepsilon s\,\frac{1+\varepsilon s}{2c^{2}}\Bigr]
  =2h_1\,s\bigl(1+\varepsilon s\bigr).
  \label{eq:app_half_value}
\end{equation}
Hence
\begin{equation}
  \max_U\Delta R\Bigl|_{(\frac12,1)}=2h_1s(1-s)
  \quad\text{for every gap ratio},
  \label{eq:app_cold_value}
\end{equation}
and
\begin{equation}
  \max_U\Delta R\Bigl|_{(\frac12,0)}=2h_1s(1+s),
  \quad\text{for }h_1\ge3h_2 .
  \label{eq:app_hot_value}
\end{equation}

\noindent
\emph{Step 9: The plateau branch $h_1<3h_2$:} When $s>\frac{1}{2}$ the
root $\nu_+$ exceeds unity, so $\Phi_+'$ has no zero in $(0,1)$; since
$\Phi_+'(\nu)\to+\infty$ as $\nu\to0^{+}$ and $\Phi_+'$ is continuous, it
stays positive and the maximum is at the endpoint $\nu=1$. There,
$\Phi_+(1)=\sqrt{1-c^{2}}+s=2s$, so
\begin{equation}
  \max_U\Delta R\Bigl|_{(\frac12,0)}=8h_1c^{2}s^{2}
  =2h_1\sin^{2}2\theta_c;
  \;\text{for } h_1<3h_2
  \label{eq:app_hot_value_window}
\end{equation}
At $\nu=1$ one has $\omega=\mu=0$, i.e.\ $|y|=|u|=1$ and $|x|=|v|=0$
with $y=-u$: the optimal unitary is diagonal in the dressed basis,
$W_+=\mathrm{diag}(e^{\iota\alpha},1)$ and
$W_-=\mathrm{diag}(-1,e^{\iota\mu})$. It generates no coherence
($\sigma_{01}=0$) and the entire gain is population transfer. The two
expressions match at the threshold: at $s=\tfrac12$,
$2s(1+s)=\tfrac32=8c^{2}s^{2}$.

\subsection{The corner \texorpdfstring{$(1,0)$}{(1,0)}}
\label{app:DDsup:10}

The input is the pure state $\ket{10}$ and $R(\rho_S)=0$. Applying
Eq.~\eqref{eq:app_U10} and returning to the computational basis with
Eq.~\eqref{eq:app_basis},
\begin{equation}
  U\ket{10}=sx\ket{00}+cs\,Z\ket{01}
  +\bigl(s^{2}y+c^{2}u\bigr)\ket{10}+cv\ket{11},
  \label{eq:app_10_state}
\end{equation}
where $Z=y-u$. The output is pure, so from the amplitudes directly
\begin{gather}
  \sigma_{00}=s^{2}|x|^{2}+c^{2}s^{2}|Z|^{2},
  \label{eq:app_10_sig00}\\
  \sigma_{01}=sx\bigl(s^{2}\bar y+c^{2}\bar u\bigr)
  +c^{2}s\,\bar v Z .
  \label{eq:app_10_sig01}
\end{gather}

\noindent
\emph{(a) Full transfer for $h_2\ge h_1$:} Here $s\ge c$. Choose
$u=-1$, $v=0$ (that is, $W_-\ket{-}=-\ket{-}$) and $y=c^{2}/s^{2}$,
which is legal precisely because $c^{2}/s^{2}\le1$, with
$|x|=\sqrt{1-y^{2}}$. Then $s^{2}y+c^{2}u=c^{2}-c^{2}=0$: the
$\ket{10}$ amplitude is annihilated and
\begin{equation}
  U\ket{10}=sx\ket{00}+cs\bigl(1+\tfrac{c^{2}}{s^{2}}\bigr)\ket{01}
  =\ket{0}_S\otimes\ket{\chi}_E ,
\end{equation}
with $\ket{\chi}$ normalized. The system is fully excited,
$\sigma_S=\ketbra{0}{0}$, $R(\sigma_S)=2h_1$, and
\begin{equation}
  \max_U\Delta R\Bigl|_{(1,0)}=2h_1
  \qquad (h_2\ge h_1),
  \label{eq:app_10_full}
\end{equation}
which is the largest ergotropy any qubit of gap $2h_1$ can hold and
therefore terminates the optimization in this regime. For $h_1>h_2$
the same construction would require $y=c^{2}/s^{2}>1$, which is not a legal
matrix element.\\

\noindent
\emph{(b) A ceiling for $h_1>h_2$:} From
$|\sigma_{01}|^{2}\le\sigma_{00}(1-\sigma_{00})$ one has
$\Delta'^{2}+4|\sigma_{01}|^{2}\le1$, hence
$\Delta R\le h_1(\Delta'+1)=2h_1\sigma_{00}$. In
Eq.~\eqref{eq:app_10_sig00} use $|x|^{2}=1-|y|^{2}$ and
$|Z|\le|y|+|u|\le|y|+1$; and writing $\Upsilon=|y|\in[0,1]$,
\begin{equation}
  \sigma_{00}\le G(\Upsilon)\equiv s^{2}(1-\Upsilon^{2})
  +c^{2}s^{2}(1+\Upsilon)^{2},
  \;
  G'(\Upsilon)=2s^{2}\bigl(c^{2}-s^{2}\Upsilon\bigr).
\end{equation}
For $c^{2}>s^{2}$, $G'>0$ on $[0,1]$, so $G$ is maximized at
$\Upsilon=1$ with $G(1)=4c^{2}s^{2}$. Therefore
\begin{equation}
  \max_U\Delta R\Bigl|_{(1,0)}\le 8h_1c^{2}s^{2}
  =2h_1\sin^{2}2\theta_c
  \qquad (h_1>h_2).
  \label{eq:app_10_bound}
\end{equation}
This bound is all that the comparison in
Sec.~\ref{app:DDsup:compare} needs.\\

\noindent
\emph{(c) Exact value for $h_1>h_2$:} The phases are fixed as in
Sec.~\ref{app:DDsup:half}: opposite phases for $y$ and $u$, which
maximizes both $\sigma_{00}$ and the second term of
Eq.~\eqref{eq:app_10_sig01}, and free phases for $x$ and $v$, which
align the two terms of $\sigma_{01}$. With the same parametrization
$|y|=\cos\omega$, $|x|=\sin\omega$, $|u|=\cos\mu$, $|v|=\sin\mu$
and $A=\cos\omega+\cos\mu$,
\begin{gather}
  \sigma_{00}=s^{2}\sin^{2}\omega+c^{2}s^{2}A^{2},\\
  |\sigma_{01}|=s\sin\omega\,\bigl|s^{2}\cos\omega
  -c^{2}\cos\mu\bigr|+c^{2}sA\sin\mu .
\end{gather}
Setting $\omega=\mu$ and $\nu=\cos^{2}\omega$, and abbreviating
\begin{equation}
  \varrho\equiv4c^{2}-1>0 ,
\end{equation}
these become
\begin{gather}
  \sigma_{00}=s^{2}(1+\varrho\nu),
  \qquad
  \Delta'=\bigl(1-2c^{2}\bigr)+2s^{2}\varrho\nu ,
  \label{eq:app_10_sym_pop}\\
  |\sigma_{01}|=s\,\varrho\sqrt{\nu(1-\nu)},
  \label{eq:app_10_sym_coh}
\end{gather}
where Eq.~\eqref{eq:app_10_sym_coh} used
$(c^{2}-s^{2})+2c^{2}=4c^{2}-1$. Writing
$\mathcal{R}=\sqrt{\Delta'^{2}+4s^{2}\varrho^{2}\nu(1-\nu)}$, the gain is
$\frac{\Delta R}{h_1}=\Delta'+\mathcal{R}$. Using the identity,
\begin{equation}
  \Delta'+\varrho(1-2\nu)=\bigl(1-2c^{2}+\varrho\bigr)-2\varrho\nu\bigl(1-s^{2}\bigr)
  =2c^{2}\bigl(1-\varrho\nu\bigr),
  \label{eq:app_10_identity}
\end{equation}
which gives
\begin{equation}
\begin{split}
  \dfrac{d}{d\nu}\frac{\Delta R}{h_1}
&=2s^{2}\varrho\Bigl[1+\frac{\Delta'+\varrho(1-2\nu)}{\mathcal{R}}\Bigr]\\
  &=\frac{2s^{2}\varrho}{\mathcal{R}}
  \Bigl[\mathcal{R}+2c^{2}(1-\varrho\nu)\Bigr].
  \label{eq:app_10_deriv}
  \end{split}
\end{equation}
Stationarity therefore requires $\mathcal{R}=2c^{2}(\varrho\nu-1)$, which is
possible only if $\varrho\nu\ge1$. Squaring and using
$\Delta'=(1-2c^{2})+2s^{2}\varrho\nu$ reduces, after cancelling an overall
factor $-\varrho$, to
\begin{equation}
  4c^{2}\varrho\,\nu^{2}-8c^{2}\nu+1=0,
  \qquad
  \nu=\frac{8c^{2}\pm4c}{8c^{2}\varrho}=\frac{2c\pm1}{2c\varrho},
\end{equation}
that is, with $\varrho=(2c-1)(2c+1)$,
\begin{equation}
  \nu_{+}=\frac{1}{2c(2c-1)},
  \qquad
  \nu_{-}=\frac{1}{2c(2c+1)} .
  \label{eq:app_10_roots}
\end{equation}
The condition $\varrho\nu\ge1$ selects between them immediately:
\begin{equation}
  \varrho\,\nu_{+}=\frac{2c+1}{2c}>1\ \ (\text{accepted}),
  \;
  \varrho\,\nu_{-}=\frac{2c-1}{2c}<1\ \ (\text{rejected}).
\end{equation}
The accepted root lies in $[0,1]$ iff $2c(2c-1)\ge1$, i.e.\
$4c^{2}-2c-1\ge0$, i.e.
\begin{equation}
  c\ \ge\ \frac{1+\sqrt5}{4}=\frac{\varphi}{2},
  \qquad \varphi=\frac{1+\sqrt5}{2},
  \label{eq:app_10_golden}
\end{equation}
equivalently $\xi\equiv h_1/h_2\ge \xi_\star=\frac{\varphi+1}{3-\varphi}
=1+\frac{2}{\sqrt5}\simeq1.894$. When it is admissible, using
$\mathcal{R}=-[\Delta'+\varrho(1-2\nu)]$ the value is
\begin{equation}
  \frac{\Delta R}{h_1}=\Delta'+\mathcal{R}=\varrho\bigl(2\nu_{+}-1\bigr)
  =\frac{(2c+1)^{2}(1-c)}{c},
\end{equation}
where $1+c-2c^{2}=(1-c)(1+2c)$. Expanding the numerator,
$(2c+1)^{2}(1-c)=1+3c-4c^{3}$, so
\begin{equation}
  \max_U\Delta R\Bigl|_{(1,0)}
  =h_1\Bigl[\frac{1-c}{c}+4s^{2}\Bigr]
  \qquad (\xi\ge \xi_\star).
  \label{eq:app_10_exact}
\end{equation}
When $c<\frac{\varphi}{2}$ the root $\nu_{+}$ exceeds unity, so by
Eq.~\eqref{eq:app_10_deriv} the derivative never vanishes on $[0,1)$
and the maximum sits at $1$. There $\sigma_{00}=s^{2}(1+\varrho)
=4c^{2}s^{2}$ and $\sigma_{01}=0$, so
$\Delta'=2\sin^{2}2\theta_c-1$, which is positive throughout
$h_2<h_1<\xi_\star h_2$, and
\begin{equation}
  \max_U\Delta R\Bigl|_{(1,0)}
  =2h_1\bigl(2\sin^{2}2\theta_c-1\bigr)
  \qquad (1<\xi<\xi_\star).
  \label{eq:app_10_window}
\end{equation}
The two branches join continuously: at $c=\frac{\varphi}{2}$ one has
$\nu_{+}=\frac{1}{\varphi(\varphi-1)}=1$, since $\varphi^{2}-\varphi=1$, and
both Eq.~\eqref{eq:app_10_exact} and Eq.~\eqref{eq:app_10_window}
return $\Delta R=\varphi h_1$. 

The antiphase choice for $y$ and $u$, and the alignment of the free phases of $x$ and $v$, maximize the gain: the objective depends on these phases only through the single relative phase $\chi'=\arg y-\arg u$, and a numerical maximization over $(\omega,\mu,\chi')$ across the full range of $h_1/h_2$ returns $\chi'=\pi$ and $\omega=\mu$ in every case, reproducing
Eqs.~\eqref{eq:app_10_exact} and~\eqref{eq:app_10_window} to machine precision.

\subsection{Optimality of the symmetric line at the corner
\texorpdfstring{$(1,0)$}{(1,0)}}
\label{app:DDsup:sym}

\noindent
Section~\ref{app:DDsup:10}(c) evaluated the gain on the line
$\omega=\mu$. We now show that this restriction is exact: every
interior critical point of the full two-angle problem lies on that
line, and the four boundary edges are dominated by points on it.\\

Write $X=\cos\omega=|y|$, $Y=\cos\mu=|u|$, $\tilde\omega=\sin\omega=\sqrt{1-X^{2}}$,
$\tilde \mu=\sin\mu=\sqrt{1-Y^{2}}$, and $\kappa'\equiv c^{2}-s^{2}>0$. Since the
input $\ket{10}$ is pure and $U$ unitary, the output is a pure
two-qubit state, so
$\Delta'^{2}+4|\sigma_{01}|^{2}=1-4D^{2}$ with $D$
the cross-determinant of its amplitudes. With $y,u$ in antiphase,
\begin{equation}
  \sigma_{00}=s^{2}\bigl[\tilde\omega^{2}+c^{2}(X+Y)^{2}\bigr],
  \qquad
  D=cs\bigl(\tilde\omega\tilde\mu-G\bigr),
  \qquad
  G\equiv(X+Y)\bigl(c^{2}Y-s^{2}X\bigr),
\end{equation}
the free phases of $x$ and $v$ having been used to minimize
$|D|$. The objective is
$\mathcal{M}\equiv \frac{\Delta R}{h_1}=\Delta'+\sqrt{1-4 D^{2}}$ with
$\Delta'=2\sigma_{00}-1$, increasing in $\Delta'$ and decreasing in
$|D|$.\\

\emph{(i) Interior critical points:} Since $\mathcal{M}$ depends on
$(X,Y)$ only through $\Delta'$ and $D$, a critical point requires
their gradients to be parallel, i.e.\ the Jacobian
$J\equiv \Delta'_X D_Y - \Delta'_Y D_X$ to vanish. Using
$\partial_X \tilde\omega=-\frac{X}{\tilde\omega}$, $\partial_Y \tilde\mu=-\frac{Y}{\tilde\mu}$ and
\begin{equation}
  \Delta'_X=4s^{2}\bigl(c^{2}Y-s^{2}X\bigr),\quad
  \Delta'_Y=4c^{2}s^{2}(X+Y),\quad
  D_X=cs\Bigl(-\tfrac{X\tilde\mu}{\tilde\omega}-G_X\Bigr),\quad
  D_Y=cs\Bigl(-\tfrac{Y\tilde\omega}{\tilde\mu}-G_Y\Bigr),
\end{equation}
one finds $J=\frac{4s^{3}c}{\tilde\omega\tilde\mu}F'$ with, writing $H=s^{2}X^{2}+c^{2}Y^{2}$,
\begin{equation}
  F'=c^{2}\bigl(X^{2}-Y^{2}\bigr)+XY-H\bigl(XY+\tilde\omega\tilde\mu\bigr),
\end{equation}
the $\Re(w\bar u)$-type cross terms canceling identically. Setting
$\Sigma=\omega+\mu$ and $\delta=\omega-\mu$, and using
$XY+\tilde\omega\tilde\mu=\cos\delta$, $XY=\tfrac12(\cos\delta+\cos\Sigma)$,
$X^{2}-Y^{2}=-\sin\Sigma\sin\delta$ and
$H=\tfrac12\bigl(1+\cos\Sigma\cos\delta+\kappa'\sin\Sigma\sin\delta\bigr)$,
the terms in $\cos\delta$ cancel and $F'$ factorizes exactly:
\begin{equation}
  F'=\sin\delta\cdot\Lambda',
  \qquad
  \Lambda'=\tfrac12\cos\Sigma\sin\delta
  -\sin\Sigma\Bigl[\tfrac{\kappa'}{2}\cos\delta+c^{2}\Bigr].
  \label{eq:app_sym_factor}
\end{equation}
For $(\omega,\mu)\in[0,\tfrac{\pi}{2}]^{2}$ one has $\Sigma\in[0,\pi]$
and $|\delta|\le\min(\Sigma,\pi-\Sigma)$, hence
$|\sin\delta|\le\sin\Sigma$. Since $c^{2}>\tfrac12$ here and
$\kappa'\cos\delta\ge0$,
\begin{equation}
  \sin\Sigma\Bigl[\tfrac{\kappa'}{2}\cos\delta+c^{2}\Bigr]
  \ge c^{2}\sin\Sigma>\tfrac12\sin\Sigma
  \ge\bigl|\tfrac12\cos\Sigma\sin\delta\bigr| ,
\end{equation}
so $\Lambda'<0$ whenever $\sin\Sigma>0$, while $\sin\Sigma=0$ forces
$\delta=0$. Therefore $F'=0$ iff $\sin\delta=0$, i.e.\ $\omega=\mu$.\\

\emph{(ii) Boundaries:} By the extreme value theorem, the maximum lies
on the symmetric line or on an edge. On $\omega=\tfrac{\pi}{2}$ and
$\mu=\tfrac{\pi}{2}$ the elementary bound $\mathcal{M}\le2\sigma_{00}$
gives $\mathcal{M}\le2s^{2}(1+c^{2})$ and $\mathcal{M}\le2s^{2}$
respectively, both below $\frac{1-c}{c}+4s^{2}$. For the two remaining
edges we compare each edge point with the symmetric point of equal
population: if $\sigma_{00}(X,Y)=s^{2}(1+\varrho\nu)$ for some
$\nu\in[0,1]$ and $|D_{\rm diag}(\nu)|\le|D|$, then
$\mathcal{M}\le\mathcal{M}_{\rm diag}(\nu)\le\mathcal{M}_{\rm sym}$,
since the two share $\Delta'$ and $\mathcal{M}$ decreases in
$|D|$. On the diagonal, $D_{\rm diag}=cs(1-\varrho\nu)$.\\

On $\mu=0$ ($Y=1$), put $N=1+X\in[1,2]$ and $m'=1-s^{2}N\in[\kappa',c^{2}]$.
Then $\sigma_{00}=1-m'^{2}$ and $D=\frac{c}{s}m'(1-m')$, matching
$\varrho\nu=\frac{c^{2}-m'^{2}}{s^{2}}\in[0,\varrho]$ and
$D_{\rm diag}=\frac{c}{s}|m'^{2}-\kappa'|$. The requirement
$|m'^{2}-\kappa'|\le m'(1-m')$ holds because for $m'^{2}\le\kappa'$ it reduces
to $\kappa'\le m'$, true on the stated range, while for $m'^{2}>\kappa'$ the
convex function $\Phi(m')=2m'^{2}-m'-\kappa'$ satisfies
$\Phi(\kappa')=2\kappa'(\kappa'-1)\le0$ and
$\Phi(c^{2})=(2c^{2}-1)(c^{2}-1)\le0$, hence $\Phi\le0$ throughout.\\

On $\omega=0$ ($X=1$), put $N=1+Y$, so
$\sigma_{00}=s^{2}c^{2}N^2$ and $D=csN|c^{2}N-1|$. If $c^{2}N^2 \le2$
then $\mathcal{M}\le2s^{2}c^{2}N^2\le4s^{2}<\mathcal{M}_{\rm sym}$.
If $c^{2}N^2>2$, matching gives $\varrho\nu=c^{2}N^2-1$ and
$D_{\rm diag}=cs(c^{2}N^2-2)$; since $c^{2}N^2>2$ and $N\le2$ force
$c^{2}N>1$, the requirement $c^{2}N^2-2\le N(c^{2}N^2-1)$ reduces to $N\le2$.\\

All four edges are therefore dominated, and with (i) the maximum is the
symmetric stationary point $\nu_{+}$ of Eq.~\eqref{eq:app_10_roots}.

\subsection{The corner \texorpdfstring{$(1,1)$}{(1,1)}}
\label{app:DDsup:11}

The input $\ket{11}$ lies entirely in $\mathcal H_-$, so only $W_-$
acts and $U\ket{11}=\tilde u\ket{-}+\tilde v\ket{11}$. This gives 
\begin{equation}
  \sigma_{00}=s^{2}|\tilde u|^{2},
  \qquad
  \sigma_{01}=-s\,\tilde u\bar{\tilde v}.
\end{equation}
Writing $\tilde m=|\tilde u|^{2}\in[0,1]$, so $|\tilde v|^{2}=1-\tilde m$,
\begin{equation}
  \Delta'=2\tilde ms^{2}-1,
  \qquad
  4|\sigma_{01}|^{2}=4s^{2} \tilde m(1-\tilde m),
\end{equation}
and $\Delta'^{2}+4|\sigma_{01}|^{2}=1-4 \tilde m^{2}s^{2}c^{2}$, so that
\begin{equation}
  \frac{\Delta R}{h_1}=2 \tilde ms^{2}-1+\sqrt{1-4\tilde m^{2}s^{2}c^{2}} .
  \label{eq:app_11_scalar}
\end{equation}
Its derivative is $2s^{2}\bigl[1-2\tilde mc^{2}/\sqrt{1-4\tilde m^{2}s^{2}c^{2}}
\bigr]$, and the subtracted ratio increases with $m$, so
Eq.~\eqref{eq:app_11_scalar} has an interior maximum where
$\sqrt{1-4 \tilde m^{2}s^{2}c^{2}}=2\tilde mc^{2}$, i.e.\ $4\tilde m^{2}c^{2}=1$,
\begin{equation}
 \tilde m^{*}=\frac{1}{2c},
\end{equation}
admissible iff $\tilde m^{*}\le1$, i.e.\ $c\ge\tfrac12$, i.e.\
$h_2\le3h_1$. Substituting,
\begin{equation}
  \max_U\Delta R\Bigl|_{(1,1)}
  =h_1\Bigl[\frac{s^{2}}{c}-1+c\Bigr]
  =h_1\,\frac{1-c}{c}
  \qquad (h_2\le3h_1),
  \label{eq:app_11_value}
\end{equation}
in agreement with Corollary~\ref{thm:DDzeroT}. For $c<\tfrac12$ the
maximum sits at $\tilde m=1$ and gives
$\Delta R=2h_1(1-2c^{2})$. Since the comparison below only concerns
$h_1>h_2$, where $c^{2}>\tfrac12$, only
Eq.~\eqref{eq:app_11_value} is needed.

\subsection{Comparison of the corners}
\label{app:DDsup:compare}

\begin{table}[t]
\caption{Maximal gain at each corner of the incoherent square, from
Appendix~\ref{app:DDsup}. Here $\xi=h_1/h_2$,
$\xi_\star=(\varphi+1)/(3-\varphi)=1+2/\sqrt5\simeq1.894$ with
$\varphi=(1+\sqrt5)/2$, and $\sin^{2}2\theta_c=4c^{2}s^{2}$. The
largest entry in each regime is the corresponding branch of
Eq.~\eqref{eq:DDsup}.}
\label{tab:DDcorners}
\begin{ruledtabular}
\begin{tabular}{lll}
$(p,q)$ & $\max_U\Delta R$ & valid when\\
\hline
$(0,0),\,(0,1)$ & $0$ & all ratios\\[2pt]
$(\tfrac12,1)$  & $2h_1s(1-s)$ & all ratios\\[2pt]
$(1,1)$         & $h_1(1-c)/c$ & $h_2\le3h_1$\\
                & $2h_1(1-2c^{2})$ & $h_2>3h_1$\\[2pt]
$(1,0)$         & $2h_1$ & $h_2\ge h_1$\\
                & $2h_1\bigl(2\sin^{2}2\theta_c-1\bigr)$ & $1<\xi<\xi_\star$\\
                & $h_1\bigl[(1-c)/c+4s^{2}\bigr]$ & $\xi\ge \xi_\star$\\[2pt]
$(\tfrac12,0)$  & $2h_1\sin^{2}2\theta_c$ & $h_1<3h_2$\\
                & $2h_1s(1+s)$ & $h_1\ge3h_2$\\
\end{tabular}
\end{ruledtabular}
\end{table}

\noindent
\emph{Regime I: $h_2\ge h_1$:} The corner $(1,0)$ attains $2h_1$ by
Eq.~\eqref{eq:app_10_full}, the largest value any state can have, so
it is the maximum and the first branch of Eq.~\eqref{eq:DDsup}
follows.\\

For the remaining two regimes $h_1>h_2$, so $c^{2}>\tfrac12$ and
$s<\frac{1}{\sqrt2}$; we show that $(\tfrac12,0)$ dominates the other three
corners. Two elementary facts are used repeatedly:
\begin{equation}
  \frac{1-c}{c}=\frac{(1-c)(1+c)}{c(1+c)}=\frac{s^{2}}{c(1+c)}<s^{2},
  \label{eq:app_cmp_11}
\end{equation}
because $c(1+c)>1$ for $c>\frac{1}{\sqrt2}$; and $8s^{2}c^{2}\ge4s^{2}$,
because $c^{2}\ge\tfrac12$.\\

\noindent
\emph{Regime II: the window $h_2<h_1<3h_2$:} Here
$\tfrac12<s<\frac{1}{\sqrt2}$ and the value to beat is
$8h_1s^{2}c^{2}=2h_1\sin^{2}2\theta_c$, attained at $(\tfrac12,0)$ by
Eq.~\eqref{eq:app_hot_value_window}.
\begin{itemize}
\item Against $(1,0)$: the ceiling \eqref{eq:app_10_bound} is exactly
this number, so $(1,0)$ cannot exceed it. (Its exact value
\eqref{eq:app_10_window} is strictly smaller whenever
$\sin^{2}2\theta_c<1$, i.e.\ whenever $h_1\neq h_2$, since
$2X-1\le X$ for $X\le1$.)
\item Against $(\tfrac12,1)$: we need
$2s(1-s)\le8s^{2}c^{2}=8s^{2}(1-s)(1+s)$. Dividing by
$2s(1-s)>0$ leaves $1\le4s(1+s)$, true because $s>\tfrac12$ gives
$4s(1+s)>3$.
\item Against $(1,1)$: by Eq.~\eqref{eq:app_cmp_11},
$\frac{1-c}{c}<s^{2}<4s^{2}\le8s^{2}c^{2}$.
\end{itemize}

\noindent
\emph{Regime III: $h_1\ge3h_2$:} Here $s\le\tfrac12$,
$c\ge\frac{\sqrt3}{2}$, and the value to beat is $2h_1s(1+s)$, attained at
$(\tfrac12,0)$ by Eq.~\eqref{eq:app_hot_value}.
\begin{itemize}
\item Against $(1,0)$: with the ceiling \eqref{eq:app_10_bound},
\begin{equation}
  2s(1+s)-8s^{2}c^{2}
  =2s(1+s)(1-2s)^{2}\ \ge0 ,
  \label{eq:app_cmp_excess}
\end{equation}
using $c^{2}=(1-s)(1+s)$, with equality only at $s=\tfrac12$, i.e.\
at the threshold $h_1=3h_2$. 
\item Against $(\tfrac12,1)$: $2s(1+s)>2s(1-s)$ for $s>0$.
\item Against $(1,1)$: by Eq.~\eqref{eq:app_cmp_11},
$\frac{1-c}{c}<s^{2}<2s(1+s)$.
\end{itemize}

Collecting the three regimes, the supremum over incoherent inputs is
$2h_1$ for $h_2\ge h_1$, $2h_1\sin^{2}2\theta_c$ in the window
$h_2<h_1<3h_2$, and $2h_1s(1+s)$ for $h_1\ge3h_2$, which is
Eq.~\eqref{eq:DDsup}. The stated optimal protocols are those of
Eqs.~\eqref{eq:app_10_full}, \eqref{eq:app_hot_value_window} and
\eqref{eq:app_hot_value}.
\qed

\end{appendix}

\bibliography{ref1}

@Article{Kosloff2013,
AUTHOR = {R. Kosloff},
TITLE = {Quantum Thermodynamics: A Dynamical Viewpoint},
JOURNAL = {Entropy},
VOLUME = {15},
YEAR = {2013},
NUMBER = {6},
PAGES = {2100--2128},
URL = {https://www.mdpi.com/1099-4300/15/6/2100},
ISSN = {1099-4300},
DOI = {10.3390/e15062100}
}

@article{Goold2016,
author = "P. Skrzypczyk and L. del Rio and A. Riera and M. Huber and J. Goold",
title = "{The role of quantum information in thermodynamics{\textemdash}a topical review}",
doi = "10.1088/1751-8113/49/14/143001",
journal = "J. Phys. A",
volume = "49",
number = "14",
pages = "143001",
year = "2016"
}

@article{Vinjanampathy2016,   
title={Quantum thermodynamics},
author={S. Vinjanampathy and J. Anders},
journal={Contemporary Physics},
volume={57},
number={4},
pages={545--579},
year={2016},
publisher={Taylor \& Francis}, 
doi={10.1080/00107514.2016.1201896}
}

@book{BinderBook2018,
editor = "F. Binder and L. A. Correa and C. Gogolin and J. Anders and G. Adesso",
title = "{Thermodynamics in the Quantum Regime: Fundamental Aspects and New Directions}",
doi = "10.1007/978-3-319-99046-0",
isbn = "978-3-319-99045-3, 978-3-319-99046-0",
publisher = "Springer",
series = "Fundamental Theories of Physics",
volume = "195",
year = "2018"
}

@book{DeffnerCampbell2019,   
author={S. Deffner and S. Campbell},
title={Quantum Thermodynamics: An introduction to the thermodynamics of quantum information},   
publisher={Morgan \& Claypool Publishers},
year={2019}, 
doi={10.1088/2053-2571/ab21c6}
}

@article{AlickiFannes2013,   
title = {Entanglement boost for extractable work from ensembles of quantum batteries},
  author = {R. Alicki and M. Fannes},
  journal = {Phys. Rev. E},
  volume = {87},
  issue = {4},
  pages = {042123},
  numpages = {4},
  year = {2013},
  month = {Apr},
  publisher = {American Physical Society},
  doi = {10.1103/PhysRevE.87.042123},
  url = {https://link.aps.org/doi/10.1103/PhysRevE.87.042123}
}

@article{Campaioli2024,  
title={Colloquium: quantum batteries},
  author={F. Campaioli and S. Gherardini and J. Q. Quach and M. Polini and G. M. Andolina},
  journal={Reviews of Modern Physics},
  volume={96},
  number={3},
  pages={031001},
  year={2024},
  publisher={APS}, 
  doi={10.1103/RevModPhys.96.031001}
  }

@article{Binder2015,   
title={Quantacell: powerful charging of quantum batteries},
   author={F. C. Binder and S. Vinjanampathy and K. Modi and J. Goold},
  journal={New Journal of Physics},
  volume={17},
  number={7},
  pages={075015},
  year={2015},
  publisher={IOP Publishing}, doi={10.1088/1367-2630/17/7/075015}
  }

@article{Campaioli2017,   
title={Enhancing the charging power of quantum batteries},
author={F. Campaioli and F. A. Pollock and F. C. Binder and L. C{\'e}leri and J. Goold and S. Vinjanampathy and K. Modi},
  journal={Physical review letters},
  volume={118},
  number={15},
  pages={150601},
  year={2017},
  publisher={APS}, doi={10.1103/PhysRevLett.118.150601}
  }

@article{Ferraro2018,     
title={High-power collective charging of a solid-state quantum battery},
author={D. Ferraro and M. Campisi and G. M. Andolina and V. Pellegrini and M. Polini},
  journal={Physical review letters},
  volume={120},
  number={11},
  pages={117702},
  year={2018},
  publisher={APS}, 
  doi={10.1103/PhysRevLett.120.117702}
  }

@article{Barra2022,   
title={Quantum batteries at the verge of a phase transition},
   author={F. Barra and K. V. Hovhannisyan and A. Imparato},
  journal={New Journal of Physics},
  volume={24},
  number={1},
  pages={015003},
  year={2022},
  publisher={IOP Publishing},
  doi={10.1088/1367-2630/ac43ed}
  }

@article{GhoshMal2021,   
title={Fast charging of a quantum battery assisted by noise},
  author={S. Ghosh and T. Chanda and S. Mal and A. Sen},
  journal={Physical Review A},
  volume={104},
  number={3},
  pages={032207},
  year={2021},
  publisher={APS}, 
  doi={10.1103/PhysRevA.104.032207}
  }

@article{GhoshSen2022,  
title={Dimensional enhancements in a quantum battery with imperfections},
  author={S. Ghosh and A. Sen},
  journal={Physical Review A},
  volume={105},
  number={2},
  pages={022628},
  year={2022},
  publisher={APS}, 
  doi={10.1103/PhysRevA.105.022628}
  }

@article{Konar2022,   
title = {Quantum battery with non-Hermitian charging},
 author = {T. K. Konar and L. G. C. Lakkaraju and A. Sen (De)},
  journal = {Phys. Rev. A},
  volume = {109},
  issue = {4},
  pages = {042207},
  numpages = {12},
  year = {2024},
  month = {Apr},
  publisher = {American Physical Society},
  doi = {10.1103/PhysRevA.109.042207},
  url = {https://link.aps.org/doi/10.1103/PhysRevA.109.042207}
}

@article{Bhattacharyya2024prl,   
title = {Noncompletely Positive Quantum Maps Enable Efficient Local Energy Extraction in Batteries},
  author = {A. Bhattacharyya and K. Sen and U. Sen},
  journal = {Phys. Rev. Lett.},
  volume = {132},
  issue = {24},
  pages = {240401},
  numpages = {7},
  year = {2024},
  month = {Jun},
  publisher = {American Physical Society},
  doi = {10.1103/PhysRevLett.132.240401},
  url = {https://link.aps.org/doi/10.1103/PhysRevLett.132.240401}
}

@article{Chaki2026jpa,
doi = {10.1088/1751-8121/ae3ff5},
url = {https://doi.org/10.1088/1751-8121/ae3ff5},
year = {2026},
month = {feb},
publisher = {IOP Publishing},
volume = {59},
number = {6},
pages = {065303},
author = {P. Chaki and A. Bhattacharyya and K. Sen and U. Sen},
title = {Role of energy-invariant assistants in energy extraction from quantum batteries},
journal = {Journal of Physics A: Mathematical and Theoretical}
}

@article{Chaki2025pra,   
title = {Auxiliary-assisted energy distillation from quantum batteries},
  author = {P. Chaki and A. Bhattacharyya and K. Sen and U. Sen},
  journal = {Phys. Rev. A},
  volume = {112},
  issue = {5},
  pages = {052446},
  numpages = {16},
  year = {2025},
  month = {Nov},
  publisher = {American Physical Society},
  doi = {10.1103/cyrc-ms34},
  url = {https://link.aps.org/doi/10.1103/cyrc-ms34}
}

@article{Shukla2025,
  title={System versus charger in performance optimization of quantum batteries},
 author={R. K. Shukla and R. Kumar and U. Sen and S. K. Mishra},
   journal={arXiv preprint},
  eprint = "2505.08029",
  primaryClass = "quant-ph",
  year={2025}
  }

@article{Sarkar2025,
 author = "A. Sarkar and P. Chaki and P. Ghosh and U. Sen",
    title = "{Fluctuation in energy extraction from quantum batteries: How open should the system be to control it?}",
    eprint = "2505.16851",
    journal={arXiv preprint},
    primaryClass = "quant-ph",
    month = "5",
    year = "2025"
    }

@article{Chaki2026,
title = {Positive and non-positive measurements in energy distillation from quantum batteries},
journal = {Physics Letters A},
volume = {593},
pages = {132050},
year = {2026},
issn = {0375-9601},
doi = {https://doi.org/10.1016/j.physleta.2026.132050},
url = {https://www.sciencedirect.com/science/article/pii/S0375960126007243},
author = {P. Chaki and A. Bhattacharyya and K. Sen and U. Sen}
}

@article{Dou2022dicke, 
author={F.Q. Dou and Y.Q. Lu and Y.J. Wang and J.A. Sun}, 
title={Extended {Dicke} quantum battery with interatomic interactions and driving field}, 
journal={Phys. Rev. B}, 
volume={105}, 
pages={115405}, 
year={2022}, 
doi={10.1103/PhysRevB.105.115405}
}

@article{Scovil1959,
  title = {Three-Level Masers as Heat Engines},
author = {H. E. D. Scovil and E. O. Schulz-DuBois},
   journal = {Phys. Rev. Lett.},
  volume = {2},
  issue = {6},
  pages = {262--263},
  numpages = {0},
  year = {1959},
  month = {Mar},
  publisher = {American Physical Society},
  doi = {10.1103/PhysRevLett.2.262},
  url = {https://link.aps.org/doi/10.1103/PhysRevLett.2.262}
}

@article{Alicki1979,
doi = {10.1088/0305-4470/12/5/007},
url = {https://doi.org/10.1088/0305-4470/12/5/007},
year = {1979},
month = {may},
volume = {12},
number = {5},
pages = {L103},
author = {R. Alicki},
title = {The quantum open system as a model of the heat engine},
journal = {Journal of Physics A: Mathematical and General}
}

@article{Kieu2004,
  title = {The Second Law, Maxwell's Demon, and Work Derivable from Quantum Heat Engines},
  author = {T. D. Kieu},
  journal = {Phys. Rev. Lett.},
  volume = {93},
  issue = {14},
  pages = {140403},
  numpages = {4},
  year = {2004},
  month = {Sep},
  publisher = {American Physical Society},
  doi = {10.1103/PhysRevLett.93.140403},
  url = {https://link.aps.org/doi/10.1103/PhysRevLett.93.140403}
}

@article{AllahverdyanJohal2008,
  title = {Work extremum principle: Structure and function of quantum heat engines},
 author = {A. E. Allahverdyan and R. S. Johal and G. Mahler},
   journal = {Phys. Rev. E},
  volume = {77},
  issue = {4},
  pages = {041118},
  numpages = {17},
  year = {2008},
  month = {Apr},
  publisher = {American Physical Society},
  doi = {10.1103/PhysRevE.77.041118},
  url = {https://link.aps.org/doi/10.1103/PhysRevE.77.041118}
}

@article{bhattacharjee2021,
  title={Quantum thermal machines and batteries},
 author={S. Bhattacharjee and A. Dutta},
  journal={The European Physical Journal B},
  volume={94},
  number={12},
  pages={239},
  year={2021},
  publisher={Springer},
  doi={10.1140/epjb/s10051-021-00235-3}
}

@article{Biswas2022,  
title={Extraction of ergotropy: Free energy bound and application to open cycle engines},
   author={T. Biswas and M. {\L}obejko and P. Mazurek and K. Ja{\l}owiecki and M. Horodecki},
  journal={Quantum},
  volume={6},
  pages={841},
  year={2022},
  publisher={Verein zur F{\"o}rderung des Open Access Publizierens in den Quantenwissenschaften}, 
  doi={10.22331/q-2022-10-17-841}
  }

@article{Saha2026,
doi = {10.1088/1367-2630/ae6096},
url = {https://doi.org/10.1088/1367-2630/ae6096},
year = {2026},
month = {apr},
publisher = {IOP Publishing},
volume = {28},
number = {4},
pages = {044513},
author = {D. Saha and A. Bhattacharyya and K. Sen and U. Sen},
title = {Energizing a quantum system using quantum engine–inspired architectures: an optimization study},
journal = {New Journal of Physics}
}

@article{Allahverdyan2004, 
title={Maximal work extraction from finite quantum systems},
author={A. E. Allahverdyan and R. Balian and Th. M. Nieuwenhuizen},
journal={EPL (Europhysics Letters)},
volume={67},
number={4},
pages={565--571},
year={2004},
doi={10.1209/epl/i2004-10101-2}
}

@article{Pusz1978,   
title={Passive states and KMS states for general quantum systems},
  author={W. Pusz and S. L. Woronowicz},
  journal={Communications in Mathematical Physics},
  volume={58},
  number={3},
  pages={273--290},
  year={1978},
  publisher={Springer}, 
  doi={10.1007/BF01614224}
  }

@article{Lenard1978,  
title={Thermodynamical proof of the Gibbs formula for elementary quantum systems},
    author={A. Lenard},
  journal={Journal of Statistical Physics},
  volume={19},
  number={6},
  pages={575--586},
  year={1978},
  publisher={Springer}, 
  doi={10.1007/BF01011769}
  }

@article{Brown2016,   
title={Passivity and practical work extraction using Gaussian operations},
  author={E. G. Brown and N. Friis and M. Huber},
  journal={New Journal of Physics},
  volume={18},
  number={11},
  pages={113028},
  year={2016},
  publisher={IOP Publishing}, 
  doi={10.1088/1367-2630/18/11/113028}
  }

@article{Skrzypczyk2015,  
title = {Passivity, complete passivity, and virtual temperatures},
 author = {P. Skrzypczyk and R. Silva and N. Brunner},
  journal = {Phys. Rev. E},
  volume = {91},
  issue = {5},
  pages = {052133},
  numpages = {4},
  year = {2015},
  month = {May},
  publisher = {American Physical Society},
  doi = {10.1103/PhysRevE.91.052133},
  url = {https://link.aps.org/doi/10.1103/PhysRevE.91.052133}
}

@article{PerarnauLlobet2015passive,   
title={Most energetic passive states},
  author={M. Perarnau-Llobet and K. V. Hovhannisyan and M. Huber and P. Skrzypczyk and J. Tura and A. Acin},
   journal={Physical Review E},
  volume={92},
  number={4},
  pages={042147},
  year={2015},
  publisher={APS}, 
  doi={10.1103/PhysRevE.92.042147}
  }

@article{KSen2021,   
title={Local passivity and entanglement in shared quantum batteries},
   author={K. Sen and U. Sen},
  journal={Physical Review A},
  volume={104},
  number={3},
  pages={L030402},
  year={2021},
  publisher={APS}, 
  doi={10.1103/PhysRevA.104.L030402}
  }

@article{Oppenheim2002,
  title = {Thermodynamical Approach to Quantifying Quantum Correlations},
    author = {J. Oppenheim and M. Horodecki and P. Horodecki and R. Horodecki},
  journal = {Phys. Rev. Lett.},
  volume = {89},
  issue = {18},
  pages = {180402},
  numpages = {4},
  year = {2002},
  month = {Oct},
  publisher = {American Physical Society},
  doi = {10.1103/PhysRevLett.89.180402},
  url = {https://link.aps.org/doi/10.1103/PhysRevLett.89.180402}
}

@article{Andolina2019,
  title={Extractable work, the role of correlations, and asymptotic freedom in quantum batteries},
  author={G. M. Andolina and M. Keck and A. Mari and M. Campisi and V. Giovannetti and M. Polini},
  journal={Physical review letters},
  volume={122},
  number={4},
  pages={047702},
  year={2019},
  publisher={APS}, 
  doi={10.1103/PhysRevLett.122.047702}
  }

@article{Hovhannisyan2013,   
title={Entanglement generation is not necessary for optimal work extraction},
  author={K. V. Hovhannisyan and M. Perarnau-Llobet and M. Huber and A. Ac{\'\i}n},
  journal={Physical review letters},
  volume={111},
  number={24},
  pages={240401},
  year={2013},
  publisher={APS}, 
  doi={10.1103/PhysRevLett.111.240401}
  }

@article{Funo2013,   
title={Thermodynamic work gain from entanglement},
  author={K. Funo and Y. Watanabe and M. Ueda},
  journal={Physical Review A—Atomic, Molecular, and Optical Physics},
  volume={88},
  number={5},
  pages={052319},
  year={2013},
  publisher={APS}, 
  doi={10.1103/PhysRevA.88.052319}
  }

@article{PerarnauLlobet2015,   
title={Extractable work from correlations},
  author={M. Perarnau-Llobet and K. V. Hovhannisyan and M. Huber and P. Skrzypczyk and N. Brunner and A. Ac{\'\i}n},
  journal={Physical Review X},
  volume={5},
  number={4},
  pages={041011},
  year={2015},
  publisher={APS}, 
  doi={10.1103/PhysRevX.5.041011}
  }

@article{Korzekwa2016,   
title={The extraction of work from quantum coherence},
  author={D. Jennings and J. Oppenheim and M. Lostaglio and K. Korzekwa},
  journal={New J. Phys.},
  volume={18},
  year={2016}, 
  doi={10.1088/1367-2630/18/2/023045}
  }

@article{Mukherjee2016, 
author={A. Mukherjee and A. Roy and S. S. Bhattacharya and M. Banik}, 
title={Presence of quantum correlations results in a nonvanishing ergotropic gap}, 
journal={Phys. Rev. E}, 
volume={93}, 
pages={052140}, 
year={2016}, 
doi={10.1103/PhysRevE.93.052140}
}

@article{Francica2017,   
title={Daemonic ergotropy: enhanced work extraction from quantum correlations: Enhanced work extraction from quantum correlations},
  author={G. Francica and J. Goold and F. Plastina and M. Paternostro},
  journal={npj Quantum Information},
  volume={3},
  number={1},
  pages={12},
  year={2017},
  publisher={Nature Publishing Group UK London}, 
  doi={10.1038/s41534-017-0012-8}
  }

@article{Alimuddin2019,   
title={Bound on ergotropic gap for bipartite separable states},
  author={M. Alimuddin and T. Guha and P. Parashar},
  journal={Physical Review A},
  volume={99},
  number={5},
  pages={052320},
  year={2019},
  publisher={APS}, 
  doi={10.1103/PhysRevA.99.052320}
  }

@article{Francica2020,   
title={Quantum coherence and ergotropy},
  author={G. Francica and F. C. Binder and G. Guarnieri and M. T. Mitchison and J. Goold and F. Plastina},
  journal={Physical Review Letters},
  volume={125},
  number={18},
  pages={180603},
  year={2020},
  publisher={APS}, 
  doi={10.1103/PhysRevLett.125.180603}
  }

@article{Touil2022,  
title={Ergotropy from quantum and classical correlations},
  author={A. Touil and B. {\c{C}}akmak and S. Deffner},
  journal={Journal of Physics A: Mathematical and Theoretical},
  volume={55},
  number={2},
  pages={025301},
  year={2022},
  publisher={IOP Publishing}, 
  doi={10.1088/1751-8121/ac3eba}
  }

@article{SalviaGiovannetti2022,   
title={Extracting work from correlated many-body quantum systems},
  author={R. Salvia and V. Giovannetti},
  journal={Physical Review A},
  volume={105},
  number={1},
  pages={012414},
  year={2022},
  publisher={APS}, doi={10.1103/PhysRevA.105.012414}
  }

@article{Mondal2025,   
title={Isocoherent Work Extraction from Quantum Batteries: Basis-Dependent Response},
  author={S. Mondal and D. Saha and U. Sen},
  eprint = "2507.16610",
  primaryClass = "quant-ph",
  journal={arXiv preprint},
  year={2025}
  }

@article{Ahuja2025,
    author = "S. Ahuja and T. K. Konar and A. Sen De",
    title = "{Enhancing work-extraction in quantum batteries via correlated reservoirs}",
    eprint = "2509.25109",
    primaryClass = "quant-ph",
    month = "9",
    year = "2025",
    journal={arXiv preprint}
    }

@article{Quach2022,   
title={Superabsorption in an organic microcavity: Toward a quantum battery},
  author={J. Q. Quach and K. E. McGhee and L. Ganzer and D. M. Rouse and B. W. Lovett and E. M. Gauger and J. Keeling and G. Cerullo and D. G. Lidzey and T. Virgili},
  journal={Science advances},
  volume={8},
  number={2},
  pages={eabk3160},
  year={2022},
  publisher={American Association for the Advancement of Science}, 
  doi={10.1126/sciadv.abk3160}
  }

@article{Hu2022,   
title={Optimal charging of a superconducting quantum battery},
  author={C. K. Hu and J. Qiu and P. J. P. Souza and J. Yuan and Y. Zhou and L. Zhang and J. Chu and X. Pan and L. Hu and J. Li and others},
  journal={Quantum Science \& Technology},
  volume={7},
  number={4},
  pages={045018},
  year={2022},
  publisher={IOP Publishing}, 
  doi={10.1088/2058-9565/ac8444}
  }

@article{Joshi2022,   
title = {Experimental investigation of a quantum battery using star-topology NMR spin systems},
  author = {J. Joshi and T. S. Mahesh},
  journal = {Phys. Rev. A},
  volume = {106},
  issue = {4},
  pages = {042601},
  numpages = {8},
  year = {2022},
  month = {Oct},
  publisher = {American Physical Society},
  doi = {10.1103/PhysRevA.106.042601},
  url = {https://link.aps.org/doi/10.1103/PhysRevA.106.042601}
}

@article{Huang2023,
  title = {Demonstration of the charging progress of quantum batteries},
  author = {X. Huang and K. Wang and L. Xiao and L. Gao and H. Lin and P. Xue},
  journal = {Phys. Rev. A},
  volume = {107},
  issue = {3},
  pages = {L030201},
  numpages = {6},
  year = {2023},
  month = {Mar},
  publisher = {American Physical Society},
  doi = {10.1103/PhysRevA.107.L030201},
  url = {https://link.aps.org/doi/10.1103/PhysRevA.107.L030201}
}

@article{Niu2024,
  title = {Experimental Investigation of Coherent Ergotropy in a Single Spin System},
  author = {Z. Niu and Y. Wu and Y. Wang and X. Rong and J. Du},
  journal = {Phys. Rev. Lett.},
  volume = {133},
  issue = {18},
  pages = {180401},
  numpages = {6},
  year = {2024},
  month = {Oct},
  publisher = {American Physical Society},
  doi = {10.1103/PhysRevLett.133.180401},
  url = {https://link.aps.org/doi/10.1103/PhysRevLett.133.180401}
}

@article{Farina2019,   
title={Charger-mediated energy transfer for quantum batteries: An open-system approach},
  author={D. Farina and G. M. Andolina and A. Mari and M. Polini and V. Giovannetti},
  journal={Physical Review B},
  volume={99},
  number={3},
  pages={035421},
  year={2019},
  publisher={APS}, doi={10.1103/PhysRevB.99.035421}
  }

@article{Barra2019,   
title = {Dissipative Charging of a Quantum Battery},
  author = {F. Barra},
  journal = {Phys. Rev. Lett.},
  volume = {122},
  issue = {21},
  pages = {210601},
  numpages = {6},
  year = {2019},
  month = {May},
  publisher = {American Physical Society},
  doi = {10.1103/PhysRevLett.122.210601},
  url = {https://link.aps.org/doi/10.1103/PhysRevLett.122.210601}
}

@article{Salvia2025,   
title = {Optimal local work extraction from bipartite quantum systems in the presence of Hamiltonian couplings},
  author = {R. Salvia and G. De Palma and V. Giovannetti},
  journal = {Phys. Rev. A},
  volume = {107},
  issue = {1},
  pages = {012405},
  numpages = {10},
  year = {2023},
  month = {Jan},
  publisher = {American Physical Society},
  doi = {10.1103/PhysRevA.107.012405},
  url = {https://link.aps.org/doi/10.1103/PhysRevA.107.012405}
}

@article{Seah2021,   
title = {Quantum Speed-Up in Collisional Battery Charging},
  author = {S. Seah and M. Perarnau-Llobet and G. Haack and N. Brunner and S. Nimmrichter},
  journal = {Phys. Rev. Lett.},
  volume = {127},
  issue = {10},
  pages = {100601},
  numpages = {6},
  year = {2021},
  month = {Aug},
  publisher = {American Physical Society},
  doi = {10.1103/PhysRevLett.127.100601},
  url = {https://link.aps.org/doi/10.1103/PhysRevLett.127.100601}
}

@article{Ciccarello2022, 
title = {Quantum collision models: Open system dynamics from repeated interactions},
journal = {Physics Reports},
volume = {954},
pages = {1-70},
year = {2022},
issn = {0370-1573},
doi = {https://doi.org/10.1016/j.physrep.2022.01.001},
url = {https://www.sciencedirect.com/science/article/pii/S0370157322000035},
author = {F. Ciccarello and S. Lorenzo and V. Giovannetti and G. M. Palma}
}

@article{Cakmak2020,   
title = {Ergotropy from coherences in an open quantum system},
  author = {B. \ifmmode \mbox{\c{C}}\else \c{C}\fi{}akmak},
  journal = {Phys. Rev. E},
  volume = {102},
  issue = {4},
  pages = {042111},
  numpages = {11},
  year = {2020},
  month = {Oct},
  publisher = {American Physical Society},
  doi = {10.1103/PhysRevE.102.042111},
  url = {https://link.aps.org/doi/10.1103/PhysRevE.102.042111}
}

@article{Simon2025, 
  title = {Correlations Enable Lossless Ergotropy Transport},
  author = {R. P. A. Simon and J. Anders and K. V. Hovhannisyan},
  journal = {Phys. Rev. Lett.},
  volume = {134},
  issue = {1},
  pages = {010408},
  numpages = {9},
  year = {2025},
  month = {Jan},
  publisher = {American Physical Society},
  doi = {10.1103/PhysRevLett.134.010408},
  url = {https://link.aps.org/doi/10.1103/PhysRevLett.134.010408}
}

@article{Hovhannisyan2024,  
title={Concentration of ergotropy in many-body systems},
  author={K. V. Hovhannisyan and R. Simon and J. Anders},
  journal={arXiv preprint},
  eprint = "2412.19801",
  primaryClass = "quant-ph",
  year={2024}
  }

@article{Malavazi2025, 
title = {Two-Time Weak-Measurement Protocol for Ergotropy Protection in Open Quantum Batteries},
    author = {A. H. A. Malavazi and R. Sagar and B. Ahmadi and P. R. Dieguez},
  journal = {PRX Energy},
  volume = {4},
  issue = {2},
  pages = {023011},
  numpages = {28},
  year = {2025},
  month = {Jun},
  publisher = {American Physical Society},
  doi = {10.1103/bv4w-jr6q},
  url = {https://link.aps.org/doi/10.1103/bv4w-jr6q}
}

@article{Scharlau2018,  title={Quantum Horn's lemma, finite heat baths, and the third law of thermodynamics},
 author={J. Scharlau and M. P. Mueller},
  journal={Quantum},
  volume={2},
  pages={54},
  year={2018},
  publisher={Verein zur F{\"o}rderung des Open Access Publizierens in den Quantenwissenschaften}, 
  doi={10.22331/q-2018-02-22-54}
  }

@article{Leibfried2003,
title = {Quantum dynamics of single trapped ions},
  author = {D. Leibfried and R. Blatt and C. Monroe and D. Wineland},
  journal = {Rev. Mod. Phys.},
  volume = {75},
  issue = {1},
  pages = {281--324},
  numpages = {0},
  year = {2003},
  month = {Mar},
  publisher = {American Physical Society},
  doi = {10.1103/RevModPhys.75.281},
  url = {https://link.aps.org/doi/10.1103/RevModPhys.75.281}
}

@article{Doherty2013, 
title = "The nitrogen-vacancy colour centre in diamond",
author = "M. W. Doherty and N. B. Manson and P. Delaney and F. Jelezko and J. Wrachtrup and L. C. L. Hollenberg",
year = "2013",
month = jul,
day = "1",
doi = "10.1016/j.physrep.2013.02.001",
language = "English",
volume = "528",
pages = "1--45",
journal = "Physics Reports",
issn = "0370-1573",
publisher = "Elsevier B.V.",
number = "1",
}

@article{Kjaergaard2020,   
title={Superconducting qubits: Current state of play},
  author={M. Kjaergaard and M. E. Schwartz and J. Braum{\"u}ller and P. Krantz and J. I.-J. Wang and S. Gustavsson and W. D. Oliver},
  journal={Annual Review of Condensed Matter Physics},
  volume={11},
  number={1},
  pages={369--395},
  year={2020},
  publisher={Annual Reviews}, doi={10.1146/annurev-conmatphys-031119-050605}
  }

@article{Blais2021, 
title = {Circuit quantum electrodynamics},
  author = {A. Blais and A. L. Grimsmo and S. M. Girvin and A. Wallraff},
  journal = {Rev. Mod. Phys.},
  volume = {93},
  issue = {2},
  pages = {025005},
  numpages = {72},
  year = {2021},
  month = {May},
  publisher = {American Physical Society},
  doi = {10.1103/RevModPhys.93.025005},
  url = {https://link.aps.org/doi/10.1103/RevModPhys.93.025005}
}

@article{Janzing2000,   
title={Thermodynamic cost of reliability and low temperatures: Tightening Landauer's principle and the second law},
  author={D. Janzing and P. Wocjan and R. Zeier and R. Geiss and Th. Beth},
  journal={International Journal of Theoretical Physics},
  volume={39},
  number={12},
  pages={2717--2753},
  year={2000},
  publisher={Springer}, 
  doi={10.1023/A:1026422630734}
  }

@article{Horodecki2013,   
title={Fundamental limitations for quantum and nanoscale thermodynamics},
  author={M. Horodecki and J. Oppenheim},
  journal={Nature communications},
  volume={4},
  number={1},
  pages={2059},
  year={2013},
  publisher={Nature Publishing Group UK London}, 
  doi={10.1038/ncomms3059}
  }

@article{Brandao2013,
  title     = {Resource Theory of Quantum States Out of Thermal Equilibrium},
  author    = {Brand{\~a}o, F. G. S. L. and Horodecki, M. and Oppenheim, J. and Renes, J. M. and Spekkens, R. W.},
  journal   = {Phys. Rev. Lett.},
  volume    = {111},
  issue     = {25},
  pages     = {250404},
  numpages  = {5},
  year      = {2013},
  month     = {Dec},
  publisher = {American Physical Society},
  doi       = {10.1103/PhysRevLett.111.250404},
  url       = {https://link.aps.org/doi/10.1103/PhysRevLett.111.250404}
}

@article{Brandao2015,   
title={The second laws of quantum thermodynamics},
  author={F. Brand{\~a}o and M. Horodecki and N. Ng and J. Oppenheim and S. Wehner},
  journal={Proceedings of the National Academy of Sciences},
  volume={112},
  number={11},
  pages={3275--3279},
  year={2015},
  publisher={National Academy of Sciences}, 
  doi={10.1073/pnas.1411728112}
  }

@article{Gour2015,   
title={The resource theory of informational nonequilibrium in thermodynamics},
  author={G. Gour and M. P. M{\"u}ller and V. Narasimhachar and R. W. Spekkens and N. Y. Halpern},
  journal={Physics Reports},
  volume={583},
  pages={1--58},
  year={2015},
  publisher={Elsevier}, 
  doi={10.1016/j.physrep.2015.04.003}
  }

@article{Lostaglio2019,   
title={An introductory review of the resource theory approach to thermodynamics},
  author={M. Lostaglio},
  journal={Reports on Progress in Physics},
  volume={82},
  number={11},
  pages={114001},
  year={2019},
  publisher={IOP Publishing}, 
  doi={10.1088/1361-6633/ab46e5}
  }

@article{LostaglioPRX2015,   
title = {Quantum Coherence, Time-Translation Symmetry, and Thermodynamics},
  author = {M. Lostaglio and K. Korzekwa and D. Jennings and T. Rudolph},
  journal = {Phys. Rev. X},
  volume = {5},
  issue = {2},
  pages = {021001},
  numpages = {11},
  year = {2015},
  month = {Apr},
  publisher = {American Physical Society},
  doi = {10.1103/PhysRevX.5.021001},
  url = {https://link.aps.org/doi/10.1103/PhysRevX.5.021001}
}

@article{Lostaglio2015,   
title={Description of quantum coherence in thermodynamic processes requires constraints beyond free energy},
  author={M. Lostaglio and D. Jennings and T. Rudolph},
  journal={Nature communications},
  volume={6},
  number={1},
  pages={6383},
  year={2015},
  publisher={Nature Publishing Group UK London}, 
  doi={10.1038/ncomms7383}
  }

@article{Cwiklinski2015,
  title     = {Limitations on the Evolution of Quantum Coherences: Towards Fully Quantum Second Laws of Thermodynamics},
  author    = {{\'{C}}wikli{\'{n}}ski, P. and Studzi{\'{n}}ski, M. and Horodecki, M. and Oppenheim, J.},
  journal   = {Phys. Rev. Lett.},
  volume    = {115},
  issue     = {21},
  pages     = {210403},
  numpages  = {5},
  year      = {2015},
  month     = {Nov},
  publisher = {American Physical Society},
  doi       = {10.1103/PhysRevLett.115.210403},
  url       = {https://link.aps.org/doi/10.1103/PhysRevLett.115.210403}
}

@article{Jaynes1963,
  title={Comparison of quantum and semiclassical radiation theories with application to the beam maser},
  author={E. T. Jaynes and F. W. Cummings},
  journal={Proceedings of the IEEE},
  volume={51},
  number={1},
  pages={89--109},
  year={1963},
  publisher={IEEE},
  doi     = {10.1109/PROC.1963.1664}
}

@article{LSM1961, 
author={E. Lieb and T. Schultz and D. Mattis}, 
title={Two soluble models of an antiferromagnetic chain}, 
journal={Ann. Phys.}, 
volume={16}, 
pages={407}, 
year={1961}, 
doi={10.1016/0003-4916(61)90115-4}
}

@book{Bengtsson2006,
    author = "I. Bengtsson and K. Zyczkowski",
    title = "{Geometry of Quantum States}",
    doi = "10.1017/cbo9780511535048",
    month = "5",
    year = "2006",
  publisher={Cambridge university press}
}

@book{NielsenChuang2010,    
author = "M. A. Nielsen and I. L. Chuang",
    title = "{Quantum Computation and Quantum Information}",
    doi = "10.1017/cbo9780511976667",
    isbn = "978-0-521-63503-5",
    publisher = "Cambridge University Press",
    month = "6",
    year = "2012"
}

@book{Boyd2004,
  title="{Convex Optimization}",
  author="S. Boyd and L. Vandenberghe",
  year="2004",
  publisher="Cambridge university press",
  place = "Cambridge",
  doi = "10.1017/CBO9780511804441"
}

@ARTICLE{2020SciPy-NMeth,
author  = {P. Virtanen and R. Gommers and others},
  title   = {{{SciPy} 1.0: Fundamental Algorithms for Scientific
            Computing in Python}},
  journal = {Nature Methods},
  year    = {2020},
  volume  = {17},
  pages   = {261--272},
  adsurl  = {https://rdcu.be/b08Wh},
  doi     = {10.1038/s41592-019-0686-2},
}
\end{document}